\documentclass[DIV=classic,a4paper,10pt]{myart}
\KOMAoptions{DIV=last}

\usepackage{caption}
\usepackage{subcaption}

\newcommand{\norm}[1]{\left\Vert#1\right\Vert}
\newcommand{\abs}[1]{\left\vert#1\right\vert}
\newcommand{\Set}[1]{\ensuremath{ \left\{ #1 \right\} }}
\newcommand{\set}[1]{\ensuremath{ \{ #1 \} }}

\renewcommand{\mid}{\,|\,}
\newcommand{\Mid}{\:\big | \:}

\begin{document}
 
\title{Conditional Preference Orders and their Numerical Representations} 

\author[a,1]{Samuel Drapeau}
\author[b,2]{Asgar Jamneshan}

\address[a]{SAIF/(CAFR) and Mathematics Departement, Shanghai Jiao Tong University, 211 Huaihai Road, Shanghai, China 200030}
\address[b]{Konstanz University, Universit\"atsstra\ss e 10, 78464 Konstanz, Germany}

\eMail[1]{sdrapeau@saif.sjtu.edu.cn}
\eMail[2]{asgar.jamneshan@uni-konstanz.de}


\abstract{
    We provide an axiomatic system modeling conditional preference orders which is based on conditional set theory.
    Conditional numerical representations are introduced, and a conditional version of the theorems of Debreu on the existence of numerical representations is proved.
    The conditionally continuous representations follow from a conditional version of Debreu's Gap Lemma the proof of which relies on a conditional version of the axiom of choice, free of any measurable selection argument.
    We give a conditional version of the von Neumann and Morgenstern representation as well as automatic conditional continuity results, and illustrate them by examples.
}
\keyWords{Conditional Preferences, Utility Theory, Gap Lemma, von Neumann and Morgenstern}


\date{\today}
\keyAMSClassification{91B06}
\keyJELClassification{C60, D81}
\maketitle

\section{Introduction}
In decision theory, the normative framework of preference ordering classically requires the completeness axiom.
Yet, there are good reasons to question completeness as famously pointed out by \citet{aumann01}:
\begin{quotation}
	Of all the axioms of utility theory, the completeness axiom is perhaps the most questionable.
	[\ldots] For example, certain decisions that an individual is asked to make might involve highly hypothetical situations, which he will never face in real life.
	He might feel that he cannot reach an ``honest'' decision in such cases.
    Other decision problems might be extremely complex, too complex for intuitive ``insight'', and our individual might prefer to make no decision at all in these problems.
    Is it ``rational'' to force decision in such cases?
\end{quotation}

Aumann's remark, supported by empirical evidence, triggered intensive research in terms of interpretation, axiomatization and representation of general incomplete preferences, see \citep{richter1966, peleg1970, bewley2001,dubra2002, dubra2004,eliaz2006,evren2011} and the references therein.
These authors consider incompleteness either as a result of status quo, see \citet{bewley2001}, or procedural decision making, see \citet{dubra2002}, and the numerical representations are in terms of multi-utilities.
However, Aumann's quote and a correspondence with Savage \citep{aumann1987}, where he exposes the idea of state-dependent preferences, suggest that the lack of information underlying a decision making is a natural source of incompleteness.
For instance, consider the simple situation where a person has to decide between visiting a museum or going for a walk on Sunday in one month from now.
She cannot express an unequivocal preference between these two prospective situations since it depends on the knowledge of uncertain factors like the weather, availability of an accompanying person, etc.
This information-based incompleteness suggests a contingent form of completeness.
For instance, conditioned on the event ``sunny and warm day'' she prefers a walk.
In this way, a complex decision problem, provided sufficient information, leads to an ``honest'' decision.
The present work suggests a framework formalizing this idea of a contingent decision making and its quantification.

Numerous quantification instruments in finance and economics entail a conditional dimension by mapping prospective outcomes to random variables such as for instance conditional and dynamic monetary risk measures  \citep{detlefsen01,cheridito01,ck09,penner02}, conditional expected utilities and certainty equivalents, dynamic assessment indices \citep{fritelli04, drapeau2014} or recursive utilities \citep{epstein02,epstein03}.
However, few papers address the axiomatization of conditional preferences underlying these conditional quantitative instruments.
In this direction is the work of \citet{luce1971} where an event-dependent preference ordering is considered and studied.
Their approach is further refined and extended in \citet{wakker1987} and \citet{karni1993, karni1993b}.
State-wise dependency is used in \citet{kreps02,kreps03} and \citet{marinacci02} to study intertemporal preferences and a dynamic version of preferences, respectively.
Remarkable is the abstract approach by \citet{skiadas01,skiadas02}.
He provides a set of axioms modeling conditional preferences on random variables which admit a conditional Savage representation of the form
\begin{equation*}
	U\left( x \right)=E_{Q}\left[ u\left( x \right) \mid \mathcal{A} \right]
\end{equation*}
where $\mathcal{A}$ is an algebra of events representing the information, $Q$ is a subjective probability measure and $u$ is a utility index.
As in the previous works, its decision-theoretical foundation consists of a whole family of total pre-orders $\succcurlyeq^A$, one for each event $A \in \mathcal{A}$, and a consistent aggregation property in order to obtain the conditional representation.
However, the decision maker is assumed to implicitly take into account a large number
\footnote{In a five steps binary tree, $4.294.967.296$ is the cardinality of the family of total pre-orders $\succcurlyeq^A$.} of complete pre-orders.

Our axiomatic approach differs in so far as it considers a \emph{single} but possibly \emph{incomplete} preference order $\succcurlyeq$ instead of a whole family of complete preference orders.
Even if one cannot a priori decide whether $x\succcurlyeq y$ or $y\succcurlyeq x$ for any two prospective outcomes, or acts, there may exist a contingent information $A$ conditioned on which $x$ is preferable to $y$.
In this case we formally write $x|A \succcurlyeq y|A$.
The set of contingent information is modeled as an algebra of events $\mathcal{A}=(\mathcal{A},\cap,\cup,{}^c,\emptyset,\Omega)$ of a state space $\Omega$.\footnote{Conditional set theory \citep{djkk2013} allows the contingent information to be any complete Boolean algebra.}
In order to describe the conditional nature of the preference, we require that $\succcurlyeq$ interacts consistently with the information, that is,
\begin{itemize}
    \item \emph{consistency:} if $x|A \succcurlyeq y|A$ and $B\subseteq A$, then $x| B \succcurlyeq y| B$;
    \item \emph{stability:} if $x|A \succcurlyeq y|A$ and $x|B \succcurlyeq y|B$, then $x|A\cup B\succcurlyeq y|A\cup B$;
    \item \emph{local completeness:} for every two acts $x$ and $y$ there exists a non-empty event $A$ such that either $x|A \succcurlyeq y|A$ or $x|A\succcurlyeq y|A$.
\end{itemize}
These assumptions bear a certain normative appeal in view of the conditional approach that we are aiming at.
In the context of the previous example, consistency says that if the person prefers a walk over a visit to a museum whenever it is ``sunny'' or ``warm'', then a fortiori she prefers a walk if it is ``sunny''.
Stability tells that if she prefers a walk whenever it is ``sunny'' or ``rainy'', then on any day where at least one of these conditions is met she will go for a walk.
In contrast to classical preferences, we only assume a local completeness:
For any two situations she is able to meet a decision provided enough -- possibly extremely precise\footnote{Indeed, the smaller the event, the more precise in which state of the world this event may occur. The most precise event being the singleton.} -- information.
In our example, there exists a rather unlikely, but still non-trivial, coincidence of the conditions \emph{`sunny'}, \emph{`humidity between 15 and 20\%'} and \emph{`wind between 0 and 10km/h'} under which she prefers a walk to the museum.
Unlike classical completeness, the information necessary to decide between two acts $x$ and $y$ depends on the pair $(x,y)$.
Note that if the set of contingent information reduces to the trivial information $\mathcal{A}=\set{\emptyset,\Omega}$, then, as expected, a conditional preference is a classical complete preference order.
In particular, classical decision theory is a special case of the conditional one.

Observe that our approach, as \citep{luce1971, wakker1987, kreps02, skiadas01, skiadas02}, considers an \emph{exogenously} given set of informations or events as the source of incomplete decision making.
Whereas in \citep{bewley2001,dubra2002} and the related subsequent literature on incomplete preference, the incompleteness and the resulting multi-valued representations yield an endogenous information about the nature of the incompleteness.
Incompleteness there is however not in terms of an algebra of events, and therefore not specifically related to a contingent decision making.

Our approach is also not a priori \emph{dynamic} in the sense that a single algebra of available information is given for the contingent decision making.
We do not address the question of progressive learning over time as new information reveals, resulting in an update of decisions.
This incremental learning approach in decision making is investigated by \citet{kreps02}, and recently by \citet{dillenberger2014} as well as \citet{piermont2015}.\footnote{Though, \citet{dillenberger2014} consider a static approach resulting in dynamic utility valuations that are deterministic.}
In these articles, the agent learns over time and may modify her behavior according to the new information as well as her previous choice making.
However, the underlying information structure is exogenously given -- either by a fixed dynamic structure by means of a filtration or a random tree, or by the filtration generated by the consumption paths, or even by the filtration generated by the previous preference orders.
Our approach may help in these cases by considering a sequence of conditional preference orders $\succcurlyeq_0,\succcurlyeq_1,\ldots, \succcurlyeq_t,\ldots$ with respect to an increasing sequence of algebra of events $\mathcal{A}_0\subseteq \mathcal{A}_1\subseteq \cdots \subseteq \mathcal{A}_t\subseteq \cdots$ each of which for every point in time.
We can provide an axiomatic system to describe these conditional preference orders $\succcurlyeq_t$ for each given time $t$, and derive a sequence of conditional numerical representations $U_t$.
Since we only address the case of a single information structure, that is, at a fixed given time $t$, we intentionally left out the following two questions in the dynamic context.
First, whether the decision making at time $t$ is influenced by the past information, that is a Markovian versus non-Markovian decision making.
Second, the impact at time $t$ of past and eventually future decisions.
In other terms, the interdependence structure over time of these preferences and the consequences for the dynamic utility representation\footnote{For instance, time consistency, Bellman principle, weaker time consistency, etc.} in terms of time consistency.\footnote{A topic of intensive study in mathematical finance, see \citep{cheridito01,ck09,penner02,cialenco2010,drapeau2014} among others.} 

Although being intuitive, it is mathematically not obvious what is meant by a contingent prospective act $x|A$.
The formalization of which corresponds to the notion of a conditional set, introduced recently by \citet{djkk2013}.
An heuristic introduction to conditional sets is given in Section \ref{sec:condset}.
For an exhaustive mathematical presentation we refer to \citep{djkk2013}.
The formalization and properties of conditional preferences are given in Section \ref{sec:preford}.
In Section \ref{sec:02}, we address the notion of conditional numerical representation and prove a conditional version of Debreu's existence result of continuous numerical representations.
While the proof technique differs, the classical statements in decision theory translate into the conditional framework.
For instance, a conditional version of the classical representation of \citet{neumann01} is presented in Section \ref{sec:condvnm}.
The representation of Debreu requires topological assumptions that often are not met in practice.
In Section \ref{sec:moncon}, we provide conditional results that allow to extend Debreu and Rader's theorem in a more general framework, and present automatic continuity results which allow to bypass topological assumptions.
We illustrate each of these cases by examples.
These results in their classical form rely on the Gap Lemma of \citet{debreu01,Debreu03} the conditional adaptation of which does not involve any measurable selection arguments but derives from a conditional version of the axiom of choice.
Section \ref{sec:gaptheorem} is dedicated to the formulation and the proof of this conditional Gap Lemma.
In Appendix \ref{appendix:01}, we gather some technical results and most of the proofs.


\section{Conditional Sets}\label{sec:condset}
As mentioned in the introduction, we model the contingent information, conditioned on which a decision maker ranks prospective outcomes, by an algebra of events $\mathcal{A}$.
For technical reasons, we assume that it is a $\sigma$-algebra with a probability measure defined on it.
The inclusion of two events is then to be understood in the almost sure sense.\footnote{In the theory of conditional sets, any complete algebra can be considered as a source of information and what follows also holds in this slightly more general framework. Though, from an economical point of view, most standard frameworks consider finite algebras of events or $\sigma$-algebras with a probability measure on it that describe the events of null measure, that is, those events that are considered as never occurring. A $\sigma$-algebra on which sets are identified if they coincide almost surely is complete, see \citep{halmos2009, djkk2013}. If one does not want to consider probability spaces, the Borel sets of a Polish space factorised by the sets of category one is a complete Boolean algebra.}
A set $X$ -- that in the present context describes acts -- is a \emph{conditional set} of $\mathcal{A}$ if it allows for conditioning actions $A :X \to X|A$ for each event $A \in \mathcal{A}$ which satisfy a consistency and an aggregation property:
\begin{itemize}
    \item[]\textbf{Consistency:} For any two acts $x,y\in X$ and events $A\subseteq B$, if the acts $x$ and $y$ coincide conditioned on $B$, that is, $x|B=y|B$, then they also coincide conditioned on $A$, that is, $x|A=y|A$.
    \item[]\textbf{Stability:} For any two acts $x,y\in X$ and event $A\in\mathcal{A}$, there exists an act $z\in X$ such that $z$ coincides with $x$ conditioned on $A$ and with $y$ otherwise.
        We denote this element $z=x|A+y|A^c$.\footnote{Since we assume that $\mathcal{A}$ can be made complete, the concatenation property is required for any partition of events $(A_i)\subseteq \mathcal{A}$ and family of acts $(x_i)\subseteq X$, and we denote $x=\sum x_i|A_i$ the unique element such that $x$ coincides with $x_i$ conditioned on $A_i$.}
\end{itemize}
Intuitively, the action $x \mapsto x|A$ tells how acts are conditioned on the information $A$ and $X|A$ represents the acts in $x$ conditioned to $A$.
\begin{example}\label{exep:running}
    Following the example from the introduction, there are two unconditional alternatives
    \begin{equation*}
        x=\textit{`going for a walk'} \quad \text{and} \quad y=\textit{`going to the museum'},
    \end{equation*}
    and the information is reduced to a single condition $A=$\emph{`sunny'} which yields the algebra
    \begin{equation*}
        \mathcal{A}=\set{0,A,A^c,\Omega}=\set{\text{`no information'}, \text{`sunny'}, \text{`not sunny'}, \text{`full information'}}.
    \end{equation*}
    The corresponding conditional set of acts is then given by
    \begin{equation*}
        X=\Set{x, y, x|A+y|A^c,y|A+x|A^c}.
    \end{equation*}
    For instance, the act $x|A +y|A^c$ stays for going for a walk provided it is sunny and going to the museum otherwise.
\end{example}
\begin{example}\label{exep:condQ}
    The conditional rational numbers are defined as follows:
    given two rational numbers $q_1,q_2 \in \mathbb{Q}$ and an event $A \in \mathcal{A}$, let $q:=q_1 |A + q_2 |A^c$ be the conditional rational number which is $q_1$ conditioned on $A$ and $q_2$ otherwise.\footnote{More generally, given a partition of events $(A_n)\subseteq \mathcal{A}$ and a corresponding family of rationals $(q_n)\subseteq \mathbb{Q}$, define the conditional rational number $q:=\sum  q_n |A_n$ as the conditional element which has the value $q_n$ conditioned on $A_n$.}
    The set of conditional rational numbers, denoted by $\mathbf{Q}$, is a conditional set where the conditioning action is given by $q|B = q_1|A\cap B+q_2|A^c\cap B \in \mathbf{Q}|B$.
    The conditional natural numbers $\mathbf{N}$ are defined analogously.
    In analogy to the next example, $\mathbf{N}$ and $\mathbf{Q}$ correspond to the set of random variables with rational values and natural values, respectively.
\end{example}
\begin{example}\label{exep:03bis}
    Another example is the collection 
    \begin{equation*}
        L^0(\mathcal{A})=\{x:\Omega \to \mathbb{R}\text{ such that } x\text{ is }\mathcal{A}\text{-measurable}\}  
    \end{equation*}
    of random variables.
    Given a random variable $x$ and an event $A$, the conditioning of $x$ on $A$ is the restriction $x|A:A\to \mathbb{R}$, $\omega \mapsto x|A (\omega):=x(\omega)$ for $\omega \in A$.
    For any two random variables $x,y$ and an event $A$, then $z=x|A+y|A^c$ corresponds to the random variable $x1_A+y1_{A^c}$ where $1_A$ is the indicator function of the event $A$.

    In many cases, the information algebra $\mathcal{A}$ describes the exogenous information which is however only partially available to the agent for a decision making.
    For instance, the information available tomorrow to decide about random outcomes that are due in a year from now and depend on the whole information during that year.
    This can be modelled as follows.
    Given another algebra $\mathcal{B}$ with a probability measure on it and such that $\mathcal{A}\subseteq \mathcal{B}$, we can define
    \begin{equation*}
        L^1(\mathcal{B}):=\left\{ x:\Omega \to \mathbb{R}\text{ such that }x\text{ is } \mathcal{B}\text{-measurable with } E\left[ \abs{x}\mid \mathcal{A}\right] <\infty\right\}.
    \end{equation*}
    as the set of $\mathcal{B}$-measurable random variables with finite conditional expectation with respect to tomorrow's information $\mathcal{A}$.\footnote{It is in fact an $L^0$-module as studied and introduced in \citep{fillipovic2012,kupper03}}
    Inspection shows that it defines a conditional set of $\mathcal{A}$ when considering the restrictions $x|A$ for events $A$ which are in the smaller algebra $\mathcal{A}$.
\end{example}
\begin{example}\label{exep:03}
    As it concerns decision theory, lotteries -- or probability distributions -- are often used as objects for decision making.
    We define
    \begin{equation*}
        P(\mathcal{A}):=\left\{ \mu:\Omega \to P\text{ such that } \mu\text{ is }\mathcal{A}\text{-measurable} \right\}
    \end{equation*}
    where $P$ is a set of lotteries.
    A conditional lottery can be seen as a state-dependent lottery providing for each state $\omega$ a lottery $\mu(\omega,dx)$.
    Throughout, we denote by $P(\mathcal{A})$ the set of state-dependent measurable lotteries.
    Likewise random variables, it defines a conditional set where $\mu|A$ is the conditional lottery restricted to the event $A$ and for every two conditional lotteries $\mu,\nu$ and event $A$, the conditional lottery $\eta=\mu|A+\nu|A^c$ corresponds to the conditional lottery $\mu 1_A+ \nu 1_{A^c}$.

    Another typical object are Anscombe-Aumann acts which are extensions of lotteries.
    Actually the conditional set of conditional lotteries $P(\mathcal{A})$ already represents, strictly speaking, Anscombe-Aumann acts.
    However, in our context, the decision making is contingent and therefore realised with respect to the available information $\mathcal{A}$.
    In the present form, an Anscombe-Aumann act is a state-dependent lottery but measurable with respect to a larger algebra of events $\mathcal{B}$ on which the decision maker cannot make an honest decision.
    Just as $L^1(\mathcal{B})$,  the conditional set of conditional Anscombe-Aumann acts is defined as
    \begin{equation*}
        P(\mathcal{B}):=\left\{ \mu:\Omega \to P\text{ such that } \mu\text{ is } \mathcal{B}\text{-measurable} \right\}.
    \end{equation*}
\end{example}

The relation between conditional sets is described by the \emph{conditional inclusion} which is characterized by two dimensions, a classical inclusion and a conditioning:
\begin{itemize}
    \item On the one hand, every non-empty set $Y\subseteq X$ which is \emph{stable}, that is, $x|A+ y|A^c\in Y$ for every $x,y \in Y$ and $A \in \mathcal{A}$, is a conditional subset of $X$.
    \item On the other hand, $X|A$ is a conditional set but on the relative algebra $\mathcal{A}_A:=\set{B\cap A\colon B\in \mathcal{A}}$ and a subset of $X$ conditioned on $A$.
\end{itemize}
Combining the two dimensions, a conditional set $Y$ is said to be \emph{conditionally included} in $X$, denoted $Y\sqsubseteq X$, if $Y=Z|A$ for some stable $Z\subseteq X$ and a condition $A \in \mathcal{A}$.
In that case, we say that $Y$ is a conditional set \emph{``living'' on} $A$ and if we want to emphasize the condition on which this conditional set lives, we denote it $Y|A$.
The conditional inclusion is illustrated in Figure \ref{pic:inclusion}.
If $A=\emptyset$, then $X|\emptyset$ lives nowhere and in particular is conditionally contained in any conditional set, and thus is conditionally the emptyset.
The \emph{conditional powerset} 
\begin{equation*}
    \mathcal{P}(X):=\Set{Y\colon Y\sqsubseteq X}=\left\{ Y\colon Y=Z|A\text{ for some event }A\in \mathcal{A}\text{ and a stable set }Z\subseteq X \right\}
\end{equation*}
consists of the collection of all conditional subsets of $X$. 
\begin{example}
    In Example \ref{exep:running}, the set $\set{x,y}\subseteq X$ is not stable since $x|A+y|A^c \not \in \set{x,y}$.
    Hence $\set{x,y}$ is not a conditional set whereas $Z:=\set{x; y|A+ x|A^c}$ is stable, and therefore a conditional subset of $X$ living on $\Omega$.
    However, $Y:=\set{x|A, y|A}$ is a conditional subset of $X$ living on $A$.
    Indeed, $Y|A=Z|A$.
\end{example}
The conditional intersection of two conditional sets $Y,Z$ is the intersection on the largest condition $A^\ast$ on which $Y$ and $Z$ have a non-empty classical intersection as illustrated in Figure \ref{pic:intersection}.
\begin{figure}[h]
    \centering
    \begin{subfigure}[t]{0.4\textwidth}
        \includegraphics[width=\textwidth]{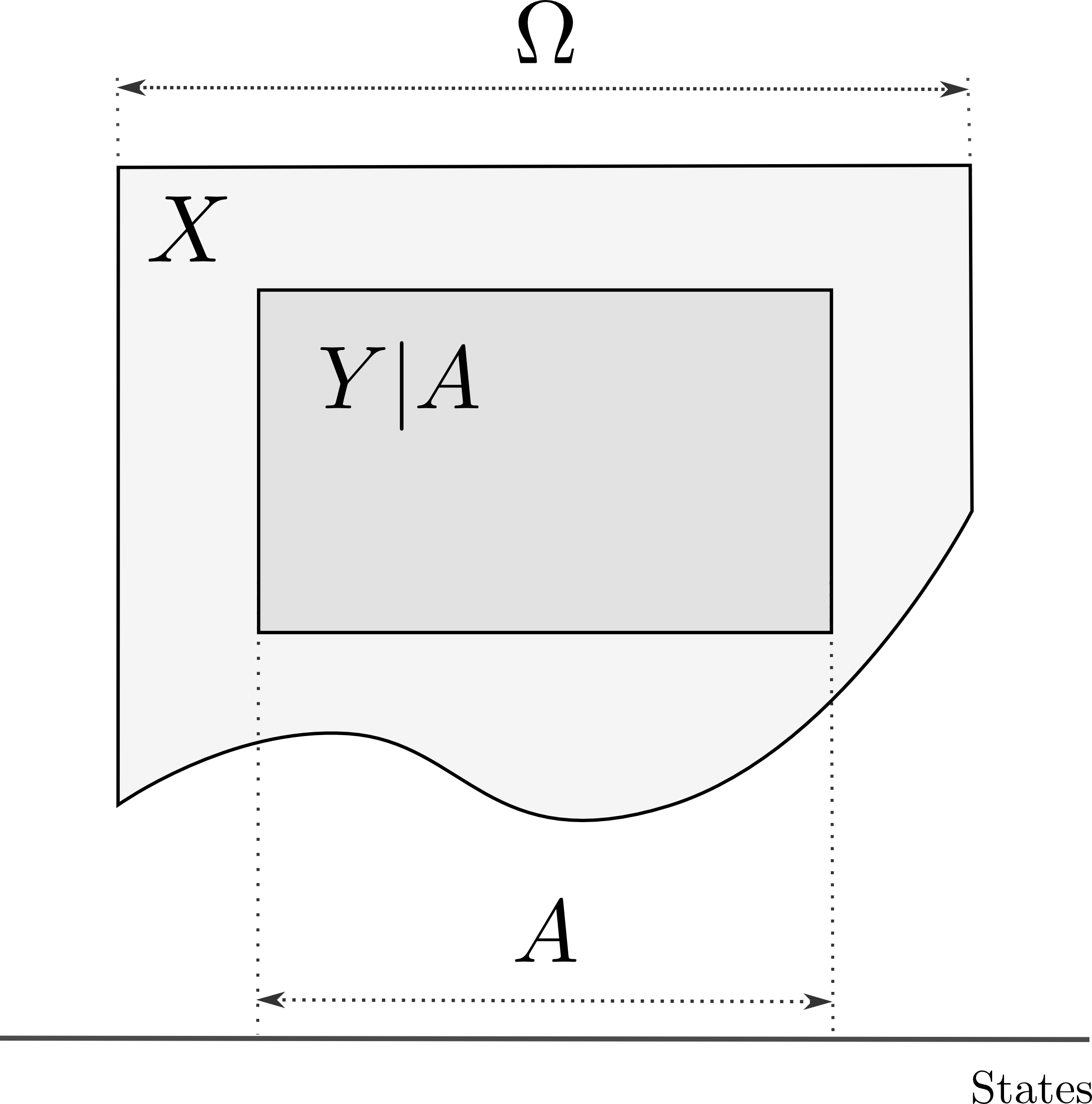}
        \caption{Illustration of the conditional inclusion.}
        \label{pic:inclusion}
    \end{subfigure}
    \qquad
    \begin{subfigure}[t]{0.4\textwidth}
        \includegraphics[width=\textwidth]{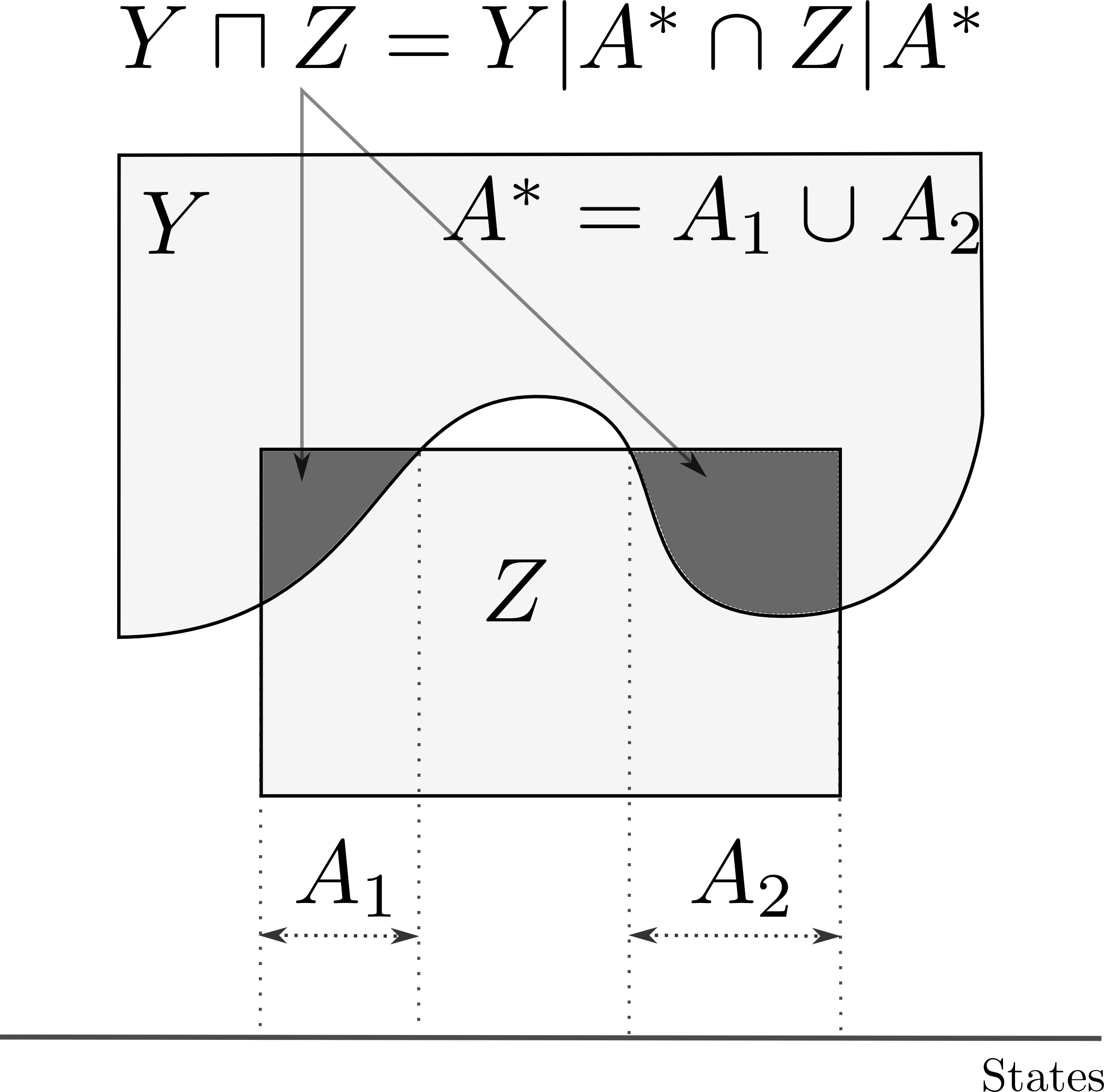}
        \caption{Illustration of the conditional intersection.}
        \label{pic:intersection}
    \end{subfigure}
\end{figure}
The conditional union of two conditional subsets $Y,Z$ is the collection of all elements which can be concatenated such that each piece of the concatenation conditionally falls either in $Y$ or in $Z$. 
The conditional union is defined by
\begin{align*}
    Y\sqcup Z:= \set{y|A+ z |B: y|A\in Y|A, z|B\in Z|B \text{ and } A\cap B=\emptyset}
\end{align*}
and illustrated in Figure \ref{pic:union}.
Finally, the conditional complement $Y^\sqsubset$ of a conditional subset $Y$ is the collection of all those elements $y$ which nowhere fall into $Y$, as illustrated in Figure \ref{pic:complement}.
\begin{figure}[h]
    \centering
    \begin{subfigure}[t]{0.4\textwidth}
        \includegraphics[width=\textwidth]{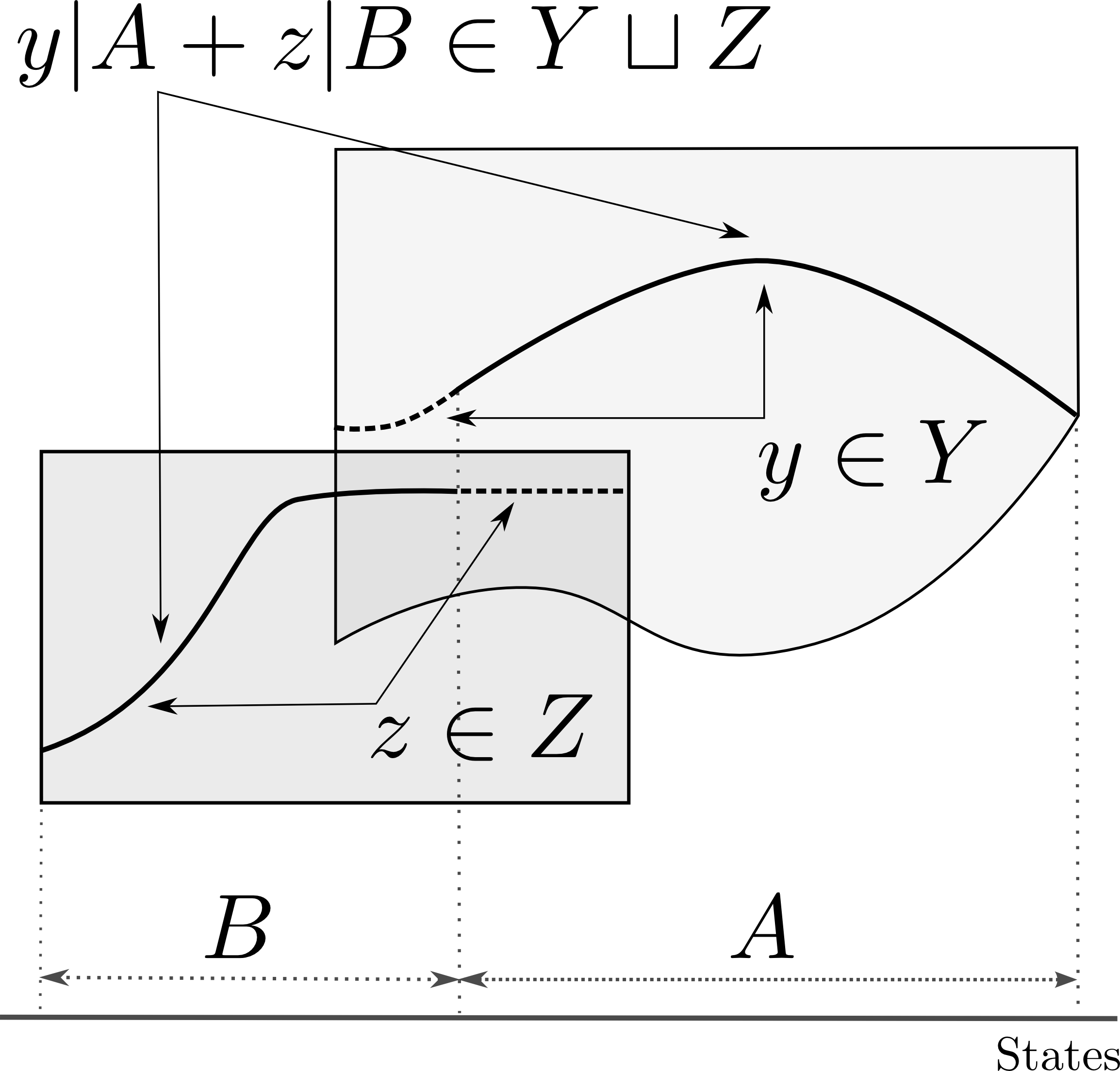}
        \caption{Illustration of the conditional union.}
        \label{pic:union}
    \end{subfigure}
    \qquad
    \begin{subfigure}[t]{0.4\textwidth}
        \includegraphics[width=\textwidth]{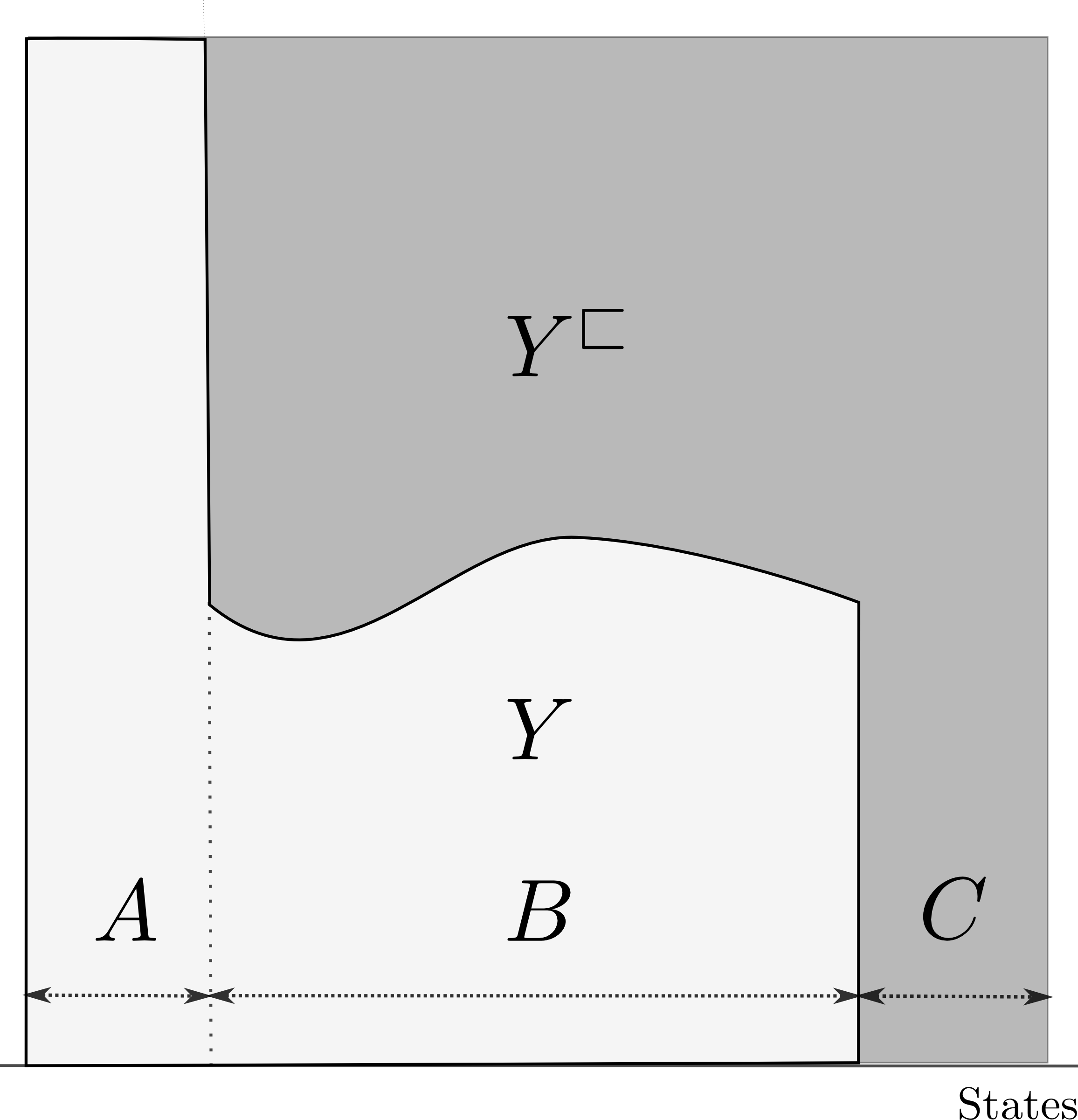}
        \caption{Illustration of the conditional complement.}
        \label{pic:complement}
    \end{subfigure}
\end{figure}

A main result in \cite{djkk2013} is that the conditional powerset together with these operations forms a complete Boolean algebra
\begin{equation*}
    \left(\mathcal{P}(X), \sqcup, \sqcap, {}^\sqsubset, X|\emptyset, X\right).
\end{equation*}
Following the classical constructions, the conditional powerset allows to define conditional relations, and functions, and topologies, and other conditional structures, see \citep{djkk2013}.

\section{Conditional Preference Orders}\label{sec:preford}
For the remainder of the paper $X$ denotes a conditional set.
A conditional binary relation $\succcurlyeq$ is a conditional subset $G\sqsubseteq X\times X$ living on $\Omega$ and we write $x\succcurlyeq y$ if and only if $(x,y) \in G$.
In particular, a conditional binary relation is at first a classical binary relation.
However, due to the fact that the graph $G$ is a conditional set and writing $x|A \succcurlyeq y|A$ for $(x|A,y|A) \in G|A$, the following additional properties hold
\begin{itemize}
    \item \emph{consistency:} if $x|A \succcurlyeq y|A$ and $B\subseteq A$, then $x|B \succcurlyeq y|B$;
    \item \emph{stability:} if $x|A \succcurlyeq y|A$ and $x|B \succcurlyeq y|B$, then $ x|A\cup B\succcurlyeq y|A\cup B$;
\end{itemize}
corresponding to two of the normative properties mentioned in the introduction.
Given a conditional binary relation, $\sim$ denotes the symmetric part of the binary relation and we use the notation 
\begin{equation*}
    x\succ y \quad \text{if and only if}\quad x\succcurlyeq y \text{ and } y|A\not\succcurlyeq x|A\text{ for every non-empty event }A\in \mathcal{A}.
\end{equation*}
In other words, $x\succ y$ means that $x$ is strictly preferred to $y$ on any non-empty condition.
Both $\sim$ and $\succ$ are conditional binary relations.
\begin{definition}
    A conditional binary relation $\succcurlyeq$ on $X$ is called a \emph{conditional preference order} if $\succcurlyeq$ is 
    \begin{itemize}
        \item \emph{reflexive:} $x\succcurlyeq x$ for every $x$;
        \item \emph{transitive:} From $x\succcurlyeq y$ and $y\succcurlyeq z$ it follows that $x\succcurlyeq z$;
        \item \emph{locally complete:} for every $x \not \sim y$ there exists a non-empty event $A$ such that either $x|A \succ y|A$ or $y|A\succ x|A$.
    \end{itemize}
\end{definition}
Although a conditional preference is not total, the following lemma shows that local completeness allows to derive for every two elements a partition on which a comparison can be achieved.
\begin{lemma}\label{lem:partition}
    Let $\succcurlyeq$ be a conditional preference order on $X$ and $x,y \in X$.
    There is a pairewise disjoint family of conditions $A,B,C$ such that $A\cup B\cup C=\Omega$ and
    \begin{equation*}
        x|A \sim y|A,\quad x|B\succ y|B\quad\text{and}\quad y|C\succ x|C.
    \end{equation*}
\end{lemma}
\begin{proof}
    Let $x,y \in X$ and define\footnote{Recall that we assume that $\mathcal{A}$ is an algebra that can be factorized in such a way that it is complete with respect to the formations of unions and intersections.}
    \begin{equation*}
        A =\cup \set{\tilde{A} \in \mathcal{A}\colon x|\tilde{A}\sim y|\tilde{A}}, \quad B  =\cup \set{\tilde{B} \in \mathcal{A}\colon x|\tilde{b}\succ y|\tilde{B}}\quad \text{and}\quad C =\cup \set{\tilde{C} \in \mathcal{A}: \tilde{C}y\succ \tilde{C}x}
    \end{equation*}
    which are the largest conditions on which $x$ is conditionally equivalent, strictly better or worse than $y$, respectively.
    Due to the consistency property of conditional relations, it follows that these conditions are mutually disjoint.
    For the sake of contradiction, suppose that $D:=A\cup B\cup C \neq \Omega$.
    It follows that outside $D$, that is, conditioned on $D^c$, the element $x$ is nowhere either equivalent, strictly better or worse than $y$.
    Otherwise, this contradicts the definition of $A$, $B$ and $C$.
    Define $\tilde{y}=x|D+y|D^c$ being $x$ conditioned on $D$ and $y$ conditioned on $D^c$.
    Since $x\not \sim y$ and $\sim$ is a conditional equivalence relation, it follows that $x \not \sim \tilde{y}$, otherwise $x|D^c\sim y|D^c$ contradicting the definition of $A$.
    By local completeness, there exists a non-empty condition $E$ such that either $x|E\succ \tilde{y}|E$ or $\tilde{y}|E\succ x|E$.
    Without loss of generality, suppose that $x|E \succ \tilde{y}|E$.
    Since $\tilde{y}|D=x|D \sim x|D$ by reflexivity and consistency, it follows that $E$ is disjoint from $D$, in other words $E\subseteq D^c$.
    In particular, $\tilde{y}|E=y|E$ implying that $x|E\succ y|E$, which together with $\emptyset\neq E\subseteq D^c$ contradict the maximality of $B$.
    Thus $D=\Omega$ which ends the proof.
\end{proof}
\begin{example}\label{example011}
    Let us give a complete formal description of the example in the introduction.
    Recall from Example \ref{exep:running} that $\mathcal{A}=\set{\text{`no information'}, \text{`sunny'}, \text{`not sunny'}, \text{`full information'}}=\{\emptyset, A, A^c,\Omega\}$, and that the conditional set generated from the two unconditional choices $x=$`going for a walk' and $y=$`going to the museum' is given by $X=\Set{x, y, x|A+y|A^c,y|A+x|A^c}$.
    The conditional preference being reflexive, it trivially holds 
    \begin{align*}
        x&\succcurlyeq x,& y&\succcurlyeq y, &x|A+ y|A^c&\succcurlyeq x|A+ y|A^c,& y|A+ x|A^c&\succcurlyeq y|A+ x|A^c.
    \end{align*}
    Further, the individual prefers going for a walk if it is sunny and to the museum otherwise.
    This translates into 
    \begin{align*}
        x|A&\succ y|A  &&\text{and}& y|A^c &\succ  x|A^c.
    \end{align*}
    Since the preference is assumed to be conditional, it also holds
    \begin{align*}
        x|A+ y|A^c&\succcurlyeq y|A+x|A^c, & x&\succcurlyeq y|A+ x|A^c ,&x|A+y|A^c&\succcurlyeq y.
    \end{align*}
    For instance, the relation $x\succcurlyeq y|A+ x|A^c$ states that going for a walk is in any case better than going to the museum if it is sunny and going for a walk otherwise.
    Inspection shows that the conditional preference is indeed a transitive and reflexive conditional relation.
    As mentioned, this relation does not tell whether $x$ is preferred to $y$, that is, whether she wants to go for a walk or to the museum.
    There exists however a condition $A=$\emph{`sunny'}, such that $x|A \succcurlyeq y|A$ which shows that it is locally complete.
    In particular, the partitioning given in Lemma \ref{lem:partition} corresponds to
    \begin{align*}
        x|\emptyset &\sim y|\emptyset, &x|A&\succ y|A,& y|A^c&\succ x|A^c.
    \end{align*}

    Note also that conditional sets allow to solve the following puzzle:
    Define $Y =\set{ z\succ y}$, the set of elements which are strictly preferred to $y$.
    In a classical setting this set is empty.
    Indeed, there exists no alternative which is strictly preferred to going to the museum since conditioned on $A^c$, $y$ is maximal for the preference order.
    However, as our intuition suggests, this set should not be empty and indeed it is conditionally non-empty since $Y=\set{x|A}$.
    The importance of this fact is observed in the proof of Debreu's Theorem \ref{thm:numrep01} with the definitions of $Z^\pm$ towards the construction of a conditional numerical representation.
\end{example}
\begin{example}\label{exep03}
    In the framework of Example \ref{exep:03bis} of $\mathcal{A}$-measurable random variables, the natural partial order on $L^0(\mathcal{A})$ given by $x\geqslant y$ if and only if $x(\omega)\geq y(\omega)$ for almost all $\omega$ is an example of a conditional preference order.
    Indeed, it is consistent since if the random variable $x$ is greater than the random variable $y$ on an event $A$, then it is also the case on any event $B\subseteq A$.
    Likewise it is stable, reflexive and transitive, even asymmetric and therefore a conditional partial order.
    However, it is also locally complete.
    Indeed, if $x\neq y$, then it follows immediately that either the event $B:=\{\omega\colon x(\omega)>y(\omega)\}$ or the event $C:=\{\omega\colon x(\omega)<y(\omega)\}$ is non-empty.
    Actually, defining the events $A=\{\omega\colon x(\omega)=y(\omega)\}$, $B=\{\omega\colon x(\omega)>y(\omega)\}$ and $C:=\{\omega\colon x(\omega)<y(\omega)\}$ provide the partition of Lemma \ref{lem:partition}.

    Since the conditional rational numbers $\mathbf{Q}$ coincide with a subset of $\mathcal{A}$-measurable random variables, the same conditional total ordering in the almost sure sense can be defined on them.
\end{example}
\begin{remark}
    In general, a conditional preference can be an equivalence relation conditioned on $A$ and strictly non-trivial on $A^c$.
    However, the case of interest lives on $A^c$.
    Therefore, throughout this paper, we assume that a conditional preference is conditionally non-trivial, that is, there exists a pair $x,y \in X$ such that $x\succ y$
\end{remark}


\section{Conditional Numerical Representations}\label{sec:02}
Next we address the quantification of such a conditional ranking.
First, we need the notion of a \emph{conditional function}.
A conditional function $f:X\to Y$ between two conditional sets is a classical function with the additional property of stability:
\begin{equation*}
    f\left(x|A+y|A^c\right)=f(x)|A+f(y)|A^c.
\end{equation*}
\begin{example} 
    The $\mathcal{A}$-conditional expectation of elements of $L^1(\mathcal{B})$ introduced in Example \ref{exep:03bis}, is a conditional function.
    Indeed, for every $x,y \in L^1(\mathcal{B})$ and each $A \in \mathcal{A}$ it holds
    \begin{align*}
        f\left( x|A+y|A \right):=E\left[ 1_A x+1_{A^c}y \mid \mathcal{A} \right]=1_AE\left[x \mid\mathcal{A} \right]+1_{A^c}E\left[ y\mid \mathcal{A} \right]=f(x)|A+f(y)|A^c
    \end{align*}
    since $1_A$ is $\mathcal{A}$-measurable.
\end{example}
\begin{example}
    For $q=q_1|A+q_2|A^c$ and $r=r_1|B+ r_2|B^c$, define the conditional addition and conditional absolute value on $\mathbf{Q}$ as
    \begin{equation*}
        q+r:=
        \begin{cases}
            q_1+r_1 &\text{ on }A  \cap B\\
            q_1+r_2 &\text{ on }A  \cap B^c\\
            q_2+r_1 &\text{ on }A^c\cap B\\
            q_2+r_2 &\text{ on }A^c\cap B^c
        \end{cases}\quad \text{and}\quad \abs{q}:= \abs{q_1}|A+\abs{q_2}|A^c.
    \end{equation*}
    Together with an analogous definition for conditional multiplication these operations make $\mathbf{Q}$ a conditional totally ordered field as defined in \citep{djkk2013}.

    In particular, this allows to define on $\mathbf{Q}$ the conditional variant of the Euclidean topology on $\mathbb{Q}$ by the conditional balls
    \begin{equation*}
        B_{r}(q):=\Set{p \in \mathbf{Q}: \abs{q-p}\leqslant r}, 
    \end{equation*}
    for $q \in \mathbf{Q}$ and $r \in \mathbf{Q}_{++}:=\Set{p \in \mathbf{Q}: p>0}$.
    It behaves like the standard topology on $\mathbb{Q}$ with the additional local property:
    \begin{equation*}
        B_{r_1}(q)|A+B_{r_2}(p)|A^c:=B_{r_1|A+r_2|A^c}(q|A+p|A), 
    \end{equation*}
    for every $r_1,r_2 \in \mathbf{Q}_{++}$ and $q,p \in \mathbf{Q}_{++}$.
    In other words, a conditional neighborhood of $3$ conditioned on $A$ and $2/5$ on $A^c$ is itself a conditional neighborhood of the conditional rational $3|A+2/5|A^c$.

\end{example}

For the quantification, we secondly need a conditional analogue of the real line which allows to represent the conditional preferences.
The conditional real numbers, denoted by $\mathbf{R}$, are obtained from the conditional rational numbers by adapting Cantor's construction, that is, identifying conditional Cauchy sequences in $\mathbf{Q}$.
As in the standard theory, the conditional real numbers can be characterized as a conditional field where every bounded subset has an infimum and a supremum and which is topologically conditionally separable.
In particular, $\mathbf{Q}$ is conditionally dense in $\mathbf{R}$.
\begin{remark}
    In our context of an algebra of events, the conditional real line $\mathbf{R}$ corresponds exactly to the conditional set of random variables $L^0(\mathcal{A})$ endowed with the $L^0$-topology introduced in \citep{kupper03}, as shown in \citep{djkk2013}.
    Therefore, in the following, the reader may always think of the conditional real line as being the set of $\mathcal{A}$-measurable random variables.
    In particular, conditional numerical representations map conditional preferences to the almost sure order between $\mathcal{A}$-measurable random variables.
\end{remark}

\begin{definition}
    A \emph{conditional numerical representation} of a conditional preference order $\succcurlyeq$ on $X$ is a conditional function $U:X \to \mathbf{R}$ such that
    \begin{equation}\label{def:relation}
        x \succcurlyeq y\quad \text{if and only if}\quad U(x)\geqslant U(y).
    \end{equation}
\end{definition}
Note that every conditional function $U:X \to \mathbf{R}$ defines a conditional preference order by means of \eqref{def:relation}.
Furthermore, if $U:X\to \mathbf{R}$ is a conditional numerical representation, then $\varphi\circ U$ is a conditional numerical representation for every conditionally strictly increasing function $\varphi:\mathbf{R}\to \mathbf{R}$.
\begin{remark}
    The \emph{conditional entropic monetary utility function} studied for instance in \citep{fritelli04} as a special case of a conditional certainty equivalent and given by
    \begin{equation*}
        U(x)=\ln\left( E\left[ e^{x}\Mid \mathcal{A} \right] \right) ,\quad x \in L^1 \left( \mathcal{B}\right),
    \end{equation*}
    is a representation of a conditional preference.
    Indeed, this function is local since for every $A \in \mathcal{A}$ it holds
    \begin{multline*}
        U(x|A+y|A^c)=\ln\left( E\left[ e^{1_Ax+1_{A^c}y}\Mid \mathcal{A} \right] \right)=\ln\left( 1_A E\left[ e^{x}\Mid \mathcal{A} \right]+ 1_{A^c}E\left[ e^{y}\Mid \mathcal{A} \right]\right)\\
        =1_A\ln\left( E\left[ e^{x}\Mid \mathcal{A} \right] \right)+1_{A^c}\ln\left( E\left[ e^{y}\Mid \mathcal{A} \right] \right)=U(x)|A+U(y)|A^c.
    \end{multline*}
    The same argumentation holds for all conditional certainty equivalents, conditional/dynamic risk measures or acceptability indices mentioned in the introduction.
\end{remark}
Given a conditional preference order, we address the necessary and sufficient conditions under which a conditional numerical representation exists.
The first result is a conditional version of Debreu's statement in \citep{debreu01} and necessitates the notion of conditionally order dense.
A conditional subset $Z \sqsubseteq X$ is \emph{conditionally order dense} if for every $x,y \in X$ with $x\succ y$, there exists $z \in Z$ such that $x\succcurlyeq z\succcurlyeq y$.
The case of interest is when $Z$ is conditionally countable, that is, there exists a conditional injection $\varphi: Z \to \mathbf{Q}$.
Equivalently, $Z$ is conditionally countable if it is a conditional sequence $Z=(z_n)_{n\in \mathbf{N}}$ where $\mathbf{N}$ is the conditional natural numbers.
There exists a difference between a conditional sequence and a standard sequence:
Analogous to the classical case, a conditional sequence $(z_n)$ in $Z$ is a conditional function $f: \mathbf{N}\to Z$, $n\mapsto f(n)=z_n$.
However stability yields $z_n|A+z_m|A^c=f(n)|A+f(m)|A^c=f(n|A+m|A^c)=z_{n|A+m|A^c}$.
In other words, the sequence step $n$ conditioned on $A$ and the sequence step $m$ conditioned on $A^c$ result into the sequence step $n|A+m|A^c$.

\begin{theorem}\label{thm:numrep01}
    A conditional preference order $\succcurlyeq$ on $X$ admits a conditional numerical representation if and only if $X$ has a conditionally countable order dense subset.
\end{theorem}
\begin{proof}
    \textbf{The if-part:}
    Without loss of generality, assume $Z=\{z_n\colon n\in\mathbf{N}\}$ is a conditionally countable order dense subset of $X$ which is not conditionally finite.
    Consider now
    \begin{equation*}
        Z^+(x):=\Set{z \in Z: z \succ x}\quad \text{and}\quad Z^-(x):=\Set{ z \in Z: x\succ z}.
    \end{equation*}
    Since $\succcurlyeq$ is a conditional binary relation, $Z^+(x)$ and $Z^-(x)$ are conditional subsets of $Z$ for every $x \in X$.
    However, as mentioned in Example \ref{example011}, $Z^\pm(x)$ may both live on some event smaller than $\Omega$.\footnote{Let $A$ be the event on which $Z^+(x)$ lives. It means that there exists no $z \in Z$ such that $z$ is strictly preferred to $x$ conditioned on $A^c$. Since $Z$ is conditionally order dense, it follows that $x$ is a maximal element conditioned on $A^c$.}
    Further, $(Z^\pm(x))_{x\in X}$ is a conditional family in the conditional powerset $\mathcal{P}(Z)$, that is, $Z^\pm(x)= Z^\pm(x_1)|A+Z^\pm(x_2)|A^c$ for every $x=x_1|A+ x_2|A^c \in X$.
    Due to transitivity, 
    \begin{equation}
        x \succcurlyeq y\quad \text{implies} \quad Z^+(x)\sqsubseteq Z^+(y)\text{ and }Z^-(y)\sqsubseteq Z^-(x).
        \label{eq:01}
    \end{equation}
    It follows from conditional order denseness that $x \succ y$ implies that there is $z \in Z$ such that $x \succ z \succcurlyeq y$ on some event $A$ and $x \succcurlyeq z \succ y$ on $A^c$.
    Thus 
    \begin{equation}
        z|A \in  \left[Z^-(x)\setminus Z^-(y)\right]|A \quad \text{and} \quad z|A^c \in \left[Z^+(y)\setminus Z^+(x)\right]|A^c\text{ for some }z \in Z\text{ and }A\in \mathcal{A}.
        \label{eq:02}
    \end{equation}
    Let now $\mu$ be a strictly positive conditional measure on $Z$\footnote{For instance, define $\mu(\set{z_n})=2^{-n}:=\sum  2^{-n_k}|A_k$ for every $n=\sum  n_k|A_k\in \mathbf{N}$.}, that is, $\mu(\set{z_n})>0$ for every $n\in\mathbf{N}$.
    Define then $U(x)=\mu(Z^-(x))-\mu(Z^+(x))$ for every $x \in X$.
    Then $U$ is a conditional function since $(Z^\pm(x))_{x\in X}$ is a conditional family and $\mu$ is a conditional function.
    On the one hand, from \eqref{eq:01} and $\mu$ being conditionally increasing it follows that $x\succcurlyeq y$ implies $U(x)\geqslant U(y)$.
    On the other hand, assume that $x\succ y$ on some non-empty event $A$. 
    Without loss of generality, $A=\Omega$.
    Then from \eqref{eq:02} and $\mu$ being strictly positive it follows that $x\succ y$ yields
    \begin{align*}
        U(x)&=\mu(Z^-(x))-\mu(Z^+(x))\\
        &\geqslant \left[\mu(\set{z})+\mu(Z^-(y))\right]|A-\left[\mu(Z^+(y))-\mu(\set{z})\right]|A^c\\
        &=\mu(\set{z})+U(y)>U(y).
    \end{align*}
    From the conditional completeness of $\succcurlyeq$ it follows that $U$ is a conditional numerical representation.

    \textbf{The only if-part:}
    A conditional preference order which admits a conditional numerical representation is  conditionally complete since the conditional reals are so. 
    It holds that $Y:=Im(U)$ is a conditional subset of $\mathbf{R}$.
    Choose a conditionally countable order dense subset $I\sqsubseteq Y$ by Lemma \ref{lem:orderdense}.
    Then $Z:=U^{-1}(I)$ is conditionally countable and since $U$ is a conditional numerical representation, it is conditionally order dense.
\end{proof}
The existence of a conditionally countable order dense subset is rather of technical nature.
In the classical case, \citet{debreu01,Debreu03} and then \citet{rader01} showed that under some topological assumptions a numerical representation exists.
And even more, by means of Debreu's Gap Lemma, the existence of an upper semi-continuous or continuous representation is guaranteed.
The conditional counterparts of these results also hold, based on a conditional adaptation of Debreu's Gap Lemma in Section \ref{sec:gaptheorem}.
\begin{definition}
    Let $\succcurlyeq$ be a conditional preference order on a conditional topological space $X$.\footnote{A conditional topology is the counterpart to a classical topology but with respect to the conditional operations of union and intersection, see \citep{djkk2013}.}
    We say that $\succcurlyeq$ is \emph{conditionally upper semi-continuous} if $\mathcal{U}(x):=\set{y\in X: y\succcurlyeq x}$ is conditionally closed for every $x\in X$.
    A conditional numerical representation $U:X\to \mathbf{R}$ is said to be conditionally upper semi-continuous if $\set{x \in X: U(x)\geqslant m}$ is conditionally closed for every $m \in \mathbf{R}$.
\end{definition}

\begin{theorem}\label{thm:rader}
    Let $\succcurlyeq$ be a conditionally upper semi-continuous preference order on a conditionally second countable\footnote{The conditional topology of which is generated by a conditionally countable neighborhood base.} topological space $X$.
    Then $\succcurlyeq$ admits a conditionally upper semi-continuous numerical representation.
    In particular, if $\succcurlyeq$ is conditionally continuous\footnote{That is, $\mathcal{U}(x)=\set{y \in X\colon y\succcurlyeq x}$ and $\mathcal{U}(y)=\set{y \in X\colon x\succcurlyeq y}$ are conditionally closed for every $x \in X$.}, then it admits a conditionally continuous numerical representation.
\end{theorem}
\begin{example}
    The conditional set $L^1(\mathcal{B})$ of Example \ref{exep:03bis} is a typical framework in which Rader's Theorem applies.
    Indeed, as soon as $\mathcal{B}$ is regular enough\footnote{That is, a separable $\sigma$-algebra.}, then it follows that $L^1(\mathcal{B})$ is a conditionally second countable Banach space for the conditional norm $\norm{x}=E[\abs{x}|\mathcal{A}]$.
    Therefore, any conditionally upper semi-continuous conditional preference order $\succcurlyeq$ on the set of random outcomes in a year from now conditioned on the information tomorrow admits a numerical representation $U$ such that
    \begin{equation*}
        x\succcurlyeq y \quad \text{if and only if}\quad U(x)\geqslant U(y)\quad \text{almost surely.}
    \end{equation*}
\end{example}
\section{A Conditional von Neumann-Morgenstern Representation}\label{sec:condvnm}
A classical class of preferences are the affine ones and the resulting affine numerical representation due to \citet{neumann01}.
This representation can be carried over to the conditional case as follows.
Let $X$ be a conditionally convex subset living on $1$ of some conditional vector space.
We say that a conditional preference order on $X$ satisfies the
\begin{itemize}
    \item \textbf{conditional independence axiom:} if $x \succ y$ then $\alpha x+(1-\alpha)z \succ \alpha y+(1-\alpha)z$ for every $z \in X$ and each $\alpha \in ]0,1]$;
    \item \textbf{conditional Archimedean axiom:} if $x \succ y \succ z$ then $\alpha x + (1-\alpha) z \succ y \succ \beta x + (1-\beta) z$ for some $\alpha, \beta\in ]0,1[$.
\end{itemize}
A conditional real-valued numerical representation $U$ of $\succcurlyeq$ is \emph{conditionally affine}, if 
\begin{equation*}
    U(\alpha x + (1-\alpha)y)=\alpha U(x) + (1-\alpha)U(y),
\end{equation*}
for every $x,y\in X$ and each $\alpha \in [0,1]$.
\begin{theorem}\label{thm:affine}
    Let $\succcurlyeq$ be a conditional preference order satisfying both the conditional Archimedean and independence axioms.
    Then $\succcurlyeq$ admits a conditionally affine representation $U$.
    Moreover, if $\hat{U}$ is another conditionally affine representation, then $\hat{U}=\alpha U+ \beta$, where $\alpha>0$ and $\beta \in \mathbf{R}$.
\end{theorem}
The result of Neumann and Morgenstern goes a step forward by providing a utility index against which lotteries are ranked according to expectation.
In our context, following Example \ref{exep:03}, conditional lotteries on the real line is the conditional set
\begin{equation*}
    P(\mathcal{A}):=\left\{ \mu:\Omega \to P \text{ such that } \mu\text{ is measurable}\right\}
\end{equation*}
where $P$ denotes a set of deterministic lotteries on the real line.\footnote{These are related to conditional distribution on the conditional real line $\mathbf{R}$, see \citet{asgar2015} for the construction and definition of such conditional probability distributions.}
We endow this conditional set with the conditional weak${}^\ast$ topology generated by the conditional set of bounded functions
\begin{equation*}
    C_b(\mathcal{A}):=\left\{ f:\Omega \to C_b\text{ such that } f\text{ is measurable} \right\}
\end{equation*}
where $C_b$ is the set of continuous functions from the real line into the real line.
In other words, a function $u \in C_b(\mathcal{A})$ is a state-dependent continuous function $u(\omega,x)$, that is, a state-dependent utility index.
The conditional scalar product is given by the random variable 
\begin{equation*}
    \omega \longmapsto \langle f, \mu\rangle(\omega) =\int_{\mathbb{R}}f(\omega,x)\mu(\omega,dx), \quad \text{for almost all }\omega \in \Omega
\end{equation*}
which is an element of $\mathbf{R}=L^0(\mathcal{A})$.
In this framework, the classical representation theorem of von Neumann and Morgenstern carries over as follows.
\begin{theorem}
    Let $\succcurlyeq$ be a conditional preference order on the conditional convex set of lotteries $P(\mathcal{A})$.
    Suppose that $\succcurlyeq$ fulfills the conditional independence and Archimedean axioms and $\succcurlyeq$ is weak${}^\ast$-continuous.
    Then there exists a unique, up to strictly positive conditionally affine transformation, conditional utility function $u\in C(\mathcal{A})$ such that
    \begin{equation*}
       \mu\succcurlyeq \nu \quad \text{if and only if}\quad\int u(\omega, x)\mu(\omega, dx)\geq \int u(\omega, x)\nu(\omega, dx)\quad \text{for almost all }\omega \in \Omega.
    \end{equation*}
\end{theorem}
\begin{proof}
    This is a consequence of Theorem \ref{thm:affine} together with a conditional version of the Riez theorem, see \citep{asgar2015}.
\end{proof}

\section{Further Conditional Representations}\label{sec:moncon}
As in the classical case, the assumptions of Theorem \ref{thm:rader} are empirically as well as mathematically problematic for the following reasons
\begin{itemize}
    \item Many of the topologies of interest for practical representations are not second countable, not even metrizable;
    \item Requiring upper semi-continuity is an empirical issue in particular for non-metrizable topologies since it is not practically falsifiable.
\end{itemize}

An answer to the first point is the following proposition that relies on a conditional version of Banach-Alaoglu based on the notion of conditional compactness introduced in \citep{djkk2013}.
\begin{proposition}\label{thm:repweak}
    Let $Y$ be a conditionally separable locally convex topological vector space admitting a conditionally countable neighborhood base of $0$, and denote by $X$ its topological dual endowed with the conditional weak${}^\ast$ topology $\sigma(X,Y)$.
    Then every conditionally upper semi-continuous preference order $\succcurlyeq$ on $X$ admits an upper semi-continuous numerical representation.
\end{proposition}
\begin{example}
    Suppose that we are interested into the decision making of an agent according to tomorrow's information $\mathcal{A}$ between random outcomes in a year from now, that is measurable with respect to some wider information algebra $\mathcal{B}\supseteq \mathcal{A}$.
    Suppose further that these random outcomes are bounded, that is, belong to the following conditional set
    \begin{equation*}
        L^\infty(\mathcal{B}):=\left\{ x:\Omega \to \mathbb{R}\text{ such that } x\text{ is measurable and bounded by an }\mathcal{A}\text{-measurable random variable} \right\}.
    \end{equation*}
    From \citep{fillipovic2012} it is known that it is the conditional dual of $L^1(\mathcal{B})$ introduced in Example \ref{exep:03bis} which is conditionally separable provided that $\mathcal{B}$ is regular enough.
    It follows that if $\succcurlyeq$ is $\sigma(L^\infty(\mathcal{B}),L^1(\mathcal{B}))$-upper semi-continuous, then it admits a numerical representation $U$.
    If furthermore, $\succcurlyeq$ is 
    \begin{itemize}
        \item \emph{conditionally convex}: $x\succcurlyeq y$ implies $\alpha x+(1-\alpha)y \succcurlyeq y$ for every conditional real number $0\leqslant \alpha \leqslant 1$ which describes a form of preference of diversification;
        \item \emph{monotone}: $x\succcurlyeq y$ whenever $x\leqslant y$, which describes a preference for almost sure better outcomes;
    \end{itemize}
    then, it follows that $U$ is conditionally quasiconvex and monotone and by means of \citep{marinacci03,drapeau01}, and the conditional extension in \citep[Theorem 2.12 and Remark 2.13]{drapeau2014}, it admits the following robust representation
    \begin{equation*}
        U(x)=\inf_{Q \in \Delta}R\left( Q, E_Q[x|\mathcal{A}] \right)
    \end{equation*}
    for a unique\footnote{The characterization of which can be found in \citep[Theorem 2.12 and Remark 2.13]{drapeau2014}.} conditional risk function $R:\Delta \times \mathbf{R}\to [-\infty, \infty]$ where $\Delta$ is the set of probability measures absolutely continuous with respect to the reference measure.
\end{example}
As for the second question, we show that automatic continuity results taking advantage of monotonicity also extends to the continuous case and yields the following result.
\begin{proposition}\label{prop:aut:cont}
    Let $\succcurlyeq$ be a conditionally complete preference order on a conditional Banach space\footnote{It also holds on a Fr\'echet lattice.} $X$.
    Suppose that $\succcurlyeq$ is conditionally convex and monotone and satisfies the
    \begin{itemize}
    \item[] \textbf{Upper Archimedean Axiom:} If $x\succcurlyeq y\succ z$, then there exists $\alpha \in \mathbf{R}$ with $0<\alpha <1$ such that $y \succ \alpha x+(1-\alpha)z$.
    \end{itemize}
    Then $\succcurlyeq$ is conditionally upper semi-continuous.
\end{proposition}
Let us illustrate this automatic continuity result in the framework presented in \citep{drapeau2014}.
By $x_t,x_{t+1},\ldots,x_T$ we denote a future cumulative cash-flow stream starting from a given future time $t$.
We are interested in the study of an investor's assessment $\succcurlyeq$ of these future cash-flows but conditioned on the available information at this time $\mathcal{A}:=\mathcal{A}_t$.
The information at time $s$ is denoted by $\mathcal{A}_s$ and it holds $\mathcal{A}_s\subseteq \mathcal{A}_{s+1}$ for every $s\geq t$.
The cash-flow $x_s$ is adapted to $\mathcal{A}_s$ and is square integrable, that is $E[ \abs{x_s}^2 | \mathcal{A}_t ]<\infty$.
We denote by
\begin{equation*}
    L^2(\mathcal{A}_s,s\geq t):=\Set{x=(x_t,\ldots,x_T)\colon x_s \in L^2(\mathcal{A}_s), s=t,\ldots,T}
\end{equation*}
which is a conditionally reflexive Hilbert space, see \citep{djkk2013,fillipovic2012}.
We denote by 
\begin{itemize}
    \item $\Delta$ the set of probability measures $Q$ absolutely continuous with respect to the reference measure;
    \item $\mathcal{D}$ the set of discounting factors, that is those processes $D=(D_t,\ldots,D_T)$ where $1=D_t\geqslant D_{t+1}\geqslant\ldots\geqslant D_T\geqslant 0$ where $ D_s$ is $\mathcal{A}_{s-1}$ adapted.
\end{itemize}
It follows that
\begin{equation*}
    E_Q\left[ \sum_{k=t}^T  D_k (x_k-x_{k-1}) \Mid \mathcal{A}_t \right]
\end{equation*}
is the expected discounted value of the cash-flow stream $x_{k}-x_{k-1}$ for the discount factor $ D \in \mathcal{D}$ under the probability model $Q \in \Delta$.
For reasons discussed in \citep{drapeau2014}, we denote by $Q\otimes  D \in \Delta\otimes \mathcal{D}$ the set of those $Q, D$ with some $L^2$-integrability conditions.
\begin{proposition}\label{rep:propproc}
    Let $\succcurlyeq$ be a conditional preference order on $L^(\mathcal{A}_s,s\geq t)$.
    Suppose that it is
    \begin{itemize}
        \item \emph{convex:} $x\succcurlyeq y$ implies $\alpha x+(1-\alpha)y\succcurlyeq y$ for every conditional real number $0\leqslant \alpha \leqslant 1$;
        \item \emph{monotone:} $x\succcurlyeq y$ whenever $x_s\geqslant y_s$ for every $s\geq t$;
        \item \emph{upper Archimedean:} if $x\succcurlyeq y\succ z$, then there exists some conditional real number $0<\alpha <1$ such that $y\succ \alpha x+(1-\alpha)z$.
    \end{itemize}
    Then, $\succcurlyeq$ is upper semi-continuous and admits an upper-semi continuous numerical representation $U$ with robust representation
    \begin{equation}\label{rep:robust02}
        U(x)=\inf_{Q\otimes  D \in \Delta\otimes \mathcal{D}}R\left( Q\otimes  D; E_Q\left[ \sum_{k=t}^T  D_k (x_k-x_{k-1}) \Mid \mathcal{A}_t \right] \right)
    \end{equation}
    for a unique minimal risk function $R:\Delta\otimes \mathcal{D}\times \mathbf{R}\to [-\infty,\infty]$.
\end{proposition}
This representation shows that convex and monotone conditional assessment of future cash flows excerpts a prudent assessment of probability as well as discounting model uncertainty.
If we can ensure the existence of an upper semi-continuous, quasiconvex, conditional and monotone numerical representation $U$, then existence and uniqueness of Representation \ref{rep:robust02} is a consequence of \citep[Theorem 3.4]{drapeau2014}.
However, $L^2(\mathcal{A}_s, s\geq t)$ being a Banach space, it follows that the assumption of the propositions fulfills the assumptions of Proposition \ref{prop:aut:cont} and so we obtain the existence of an upper semi-continuous numerical representation.


\section{Conditional Gap Lemma}\label{sec:gaptheorem}

For $s,t \in \mathbf{R}$ with $s\leqslant t$ we denote $[s,t]=\set{u \in \mathbf{R}: s\leqslant u\leqslant t}$ and if $s<t$ we denote $[s,t[=\set{u \in \mathbf{R}: s\leqslant u<t}$, $]s,t]=\set{u \in \mathbf{R}: s<u\leqslant t}$, and $]s,t[=\set{u \in \mathbf{R}: s < u<t}$, all conditionally convex subsets of $\mathbf{R}$ which live on $\Omega$. 
In the conditional topology of $\mathbf{R}$ the conditionally convex subset $[s,t]$ is conditionally closed whereas $]s,t[$ is conditionally open.
A conditionally convex subset $I\sqsubseteq \mathbf{R}$ is an \emph{interval}.
Denoting by $s =\inf I$ and $t=\sup I$, an interval is generically denoted by $(s,t)$.\footnote{It is possible that $s,t$ attain $\pm \infty$ on some positive condition.}
Up to conditioning, all conditionally convex subsets of $\mathbf{R}$ are characterized as conditional intervals. 
If we assume that an interval lives on $\Omega$, convexity yields
\begin{equation}\label{eq:condint}
    (s,t)=[s,t]|A+[s,t[|B+ ]s,t]|C+]s,t[|D
\end{equation}
where $A=E\cap F$, $B=E\cap F^c$, $C=E^c\cap F$, $D=E^c\cup F^c$ and $E =\cup\set{\tilde{E}: s|\tilde{E} \in I|\tilde{E}}$ and $F=\cup\set{\tilde{F}: t|\tilde{F} \in I|\tilde{F}}$, as illustrated in Figure \ref{pic:gap}.
\begin{figure}[h]
    \centering
    \includegraphics[scale=0.12]{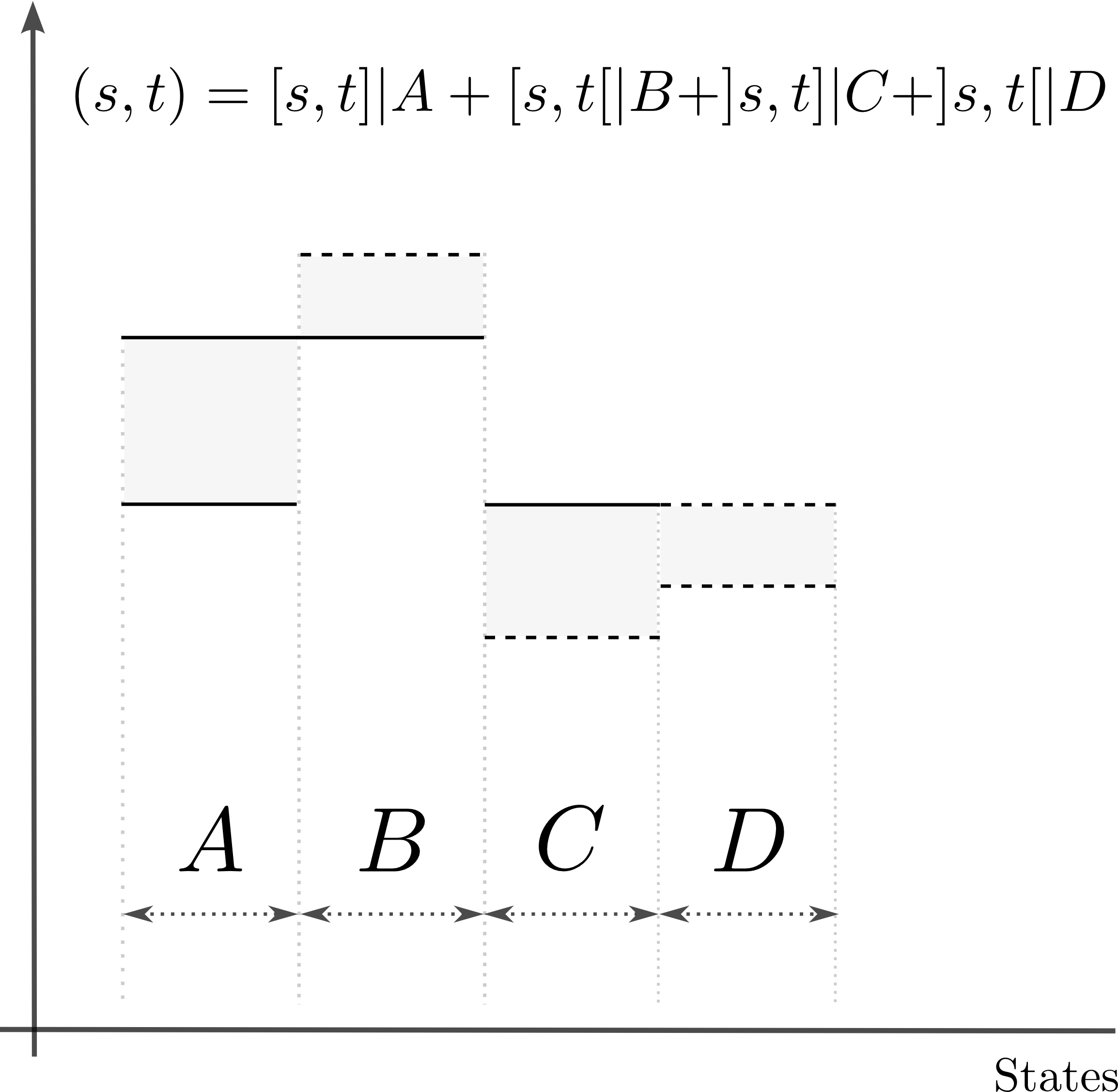}
    \caption{Illustration of a Gap.}
    \label{pic:gap}
\end{figure}

Let $S\sqsubseteq \mathbf{R}$ and $(s,t)$ be an interval\footnote{May live on an event $A$ strictly smaller than $\Omega$.} such that $(s,t)\sqsubseteq S^\sqsubset$.
Inspection shows that there exists a unique maximal interval $(s^\ast,t^\ast)$ with respect to conditional inclusion such that $(s,t)\sqsubseteq (s^\ast,t^\ast)\sqsubseteq S^\sqsubset$.
Such a maximal interval with $\inf S^\sqsubset< s^\ast \leqslant t^\ast < \sup S^\sqsubset$ on the conditions where $s^\ast,t^\ast$ are living is called a \emph{conditional gap} of $S$.
Note that any conditional gap $(s,t)$ of $S$ can be decomposed into
\begin{equation*}
    (s,t)=\set{s}|A+(s,t)|B
\end{equation*}
where $A\cap B=\emptyset$ and $s<t$ on $B$, that is the conditional interior of $(s,t)$ lives on $B$.
Moreover, the family of conditional gaps of $S$ is itself stable and therefore each of the conditional gaps of $S$ lives on the same condition.
Indeed, suppose that two conditional gaps $(s_1,t_1)$ and $(s_2,t_2)$ live on events $A$ and $B$, respectively, and $A\subseteq B$ with $A\neq B$.
Then it follows that
\begin{equation*}
    (s_1,t_1)=(s_1,t_1)|A\sqsubseteq  (s_1,t_1)|A+(s_2,t_2)|B\cap A^c \sqsubseteq S^\sqsubset
\end{equation*}
contradicting the maximality of $(s_1,t_1)$.
Hence $A=B$.
\begin{lemma}\label{lem:orderdense}
    The conditionally complete order $\geqslant$ restricted to any $S\sqsubseteq \mathbf{R}$ living on $\Omega$ admits a conditionally countable order dense subset. 
\end{lemma}
\begin{proof}
    Analogous to conditional gaps, we define a \emph{predecessor-successor} as a maximal interval $(s,t)\sqsubseteq S^\sqsubset$ but under the additional requirement that $s<t$.
    In other words, these are conditionally maximal non-trivial conditional gaps.
    Alike conditional gaps, predecessor-successor pairs of $S$ form a conditional family and therefore all live on the same condition.
    This condition is per definition smaller than the one on which the conditional gaps live.

    Now, up to conditioning, we may assume that the conditional gaps of $S$ are all living on $\Omega$.\footnote{On the condition where none of the gaps lives, it holds $S=\mathbf{R}$ for which $\mathbf{Q}$ is a conditionally countable order dense subset.}
    Since $B_{1/m}(q)$ for $m\in \mathbf{N}$ and $q \in \mathbf{Q}$ is a conditionally countable base of the topology, the family of those intersections $B_{1/m}(q)\sqcap S$ living on $\Omega$ is a conditionally countable family which we denote $(U_n)$.
    By means of the conditional axiom of choice, see \citep[Theorem 2.17]{djkk2013}, there exists a conditionally countable family $(u_n)$ such that $u_n \in U_n$ for all $n$.
    Let further $A$ be the condition on which the conditional family $(s_i,t_i)$ of the predecessors-successors of $S$ is living.
    It follows that $U=(u_n)\sqcup (s_i)$ is a conditionally countable order dense subset of $S$ living on $\Omega$.
    Indeed, let $s<t$ for $s,t \in S$ and $B$ the condition on which $(s,t)$ is a predecessor-successor pair, that is, the maximal condition such that $s$ is an element of $(s_i)$.
    It follows that there exists $v \in S$ such that $s< v< t$ on $B^c$.
    Hence, we may find $q \in \mathbf{Q}$ and $n \in \mathbf{N}$ such that $s< q-1/m<v<q+1/m<t$ on $B^c$ which ensures the existence of some $u_n$ in the family $(u_n)$ such that $s<u_n<t$ on $B^c$.
    It follows that $u=s|B+ u_n|B^c \in U$ and $s\leqslant u\leqslant t$.

    We are then left to show that $U$ is conditionally countable.
    Since $(u_n)$ is conditionally countable, according to \citep[Lemma 2.33]{djkk2013}, it is enough to show that the conditional family of open sets $]s_i,t_i[=]s_i,t_i[_{i\in I}$, where $(s_i,t_i)$ is the conditional collection of predecessor-successors, is conditionally countable.
    Without loss of generality, suppose that this family lives on $\Omega$.
    For any two $]s_i,t_i[$ and $]s_j,t_j[$ such that $s_i\neq s_j$ on any non-empty condition, it follows that $]s_i,t_i[\sqcap ]s_j,t_j[=\mathbf{R}|\emptyset$.
    This provides a conditionally pairwise disjoint family of conditionally open sets on $\Omega$.
    By means of the conditional axiom of choice, \citep[Theorem 2.17]{djkk2013}, we choose a conditional family $(q_i)$ of elements of $\mathbf{Q}$ such that $q_i \in ]s_i,t_i[$ for every $i$.
    For $P=\sqcup \set{q_i}\sqsubseteq\mathbf{Q}$, define $f:I\to P$, $i\mapsto q_i$.
    This function is a well-defined conditional function.
    Indeed, for $q_i=q_j$, it follows that $q_i \in ]s_i,t_i[\sqcap ]s_j,t_j[$.
    Both being conditional gaps of $S$, this implies that $]s_i,t_i[=]s_j,t_j[$ and therefore $i=j$.
    This also shows that $f$ is a conditional injection, thus $I$ is at most conditionally countable. 
\end{proof}

\begin{theorem}[Debreu's Gap Lemma]\label{gaptheorem}
    For every $S\sqsubseteq \mathbf{R}$ there exists a conditionally strictly increasing function $g:S \to \mathbf{R}$ such that all the conditional gaps $(s,t)$ of $g(S)$ are of the form
    \begin{equation*}
    (s,t)=\set{s}|A+]s,t[|B.
    \end{equation*}
\end{theorem}
This theorem says that there exists a strictly increasing transformation of $S$ such that any conditional gap which is of the form \eqref{eq:condint} is transformed in a gap which conditionally is either empty, a singleton or an open set.
The following argumentation follows the proof idea in \citep{ouwehand2010}.
\begin{proof}
    \begin{enumerate}[label=\textit{Step \arabic*:},fullwidth]
        \item According to Lemma \ref{lem:orderdense}, let $U=\{u_n\colon n\in\mathbf{N}\}$ be a conditionally countable order dense subset of $S$.
            We construct a conditionally increasing function $f: U\to [0,1]$.
            Let $H$ be the set of conditional functions $f:V\to [0,1]$, where $V=\set{u_k: 1\leqslant k\leqslant n}$, $n\in \mathbf{N}$ or $V=U$ and such that\footnote{With the classical convention that the infimum and supremum over emptyset is equal to $1$ and $0$, respectively. That is, $f(u_{k})=\inf_{l\leqslant k-1} f(u_l)/2$ on the condition where $u_{k}< u_l$ for every $l\leqslant k-1$ and $f(u_{k})=(\sup_{l\leqslant k-1} f(u_l)+1)/2$ on the condition where $u_{l}< u_k$ for every $k-1\leqslant l$.}
            \begin{align*}
                    f(u_1)&=1/2,\\\\
                    \displaystyle f(u_k)&=\frac{\displaystyle \sup_{l\leqslant k-1} \Set{f(u_l): u_l <u_{k}}+\inf_{l\leqslant k-1}\Set{f(u_l): u_{k} < u_l}}{2},& k\geqslant 2.
            \end{align*}
            By definition, any $f \in H$ is conditionally strictly increasing on its domain and $H$ is a conditional set.
            Furthermore $f:\set{u_1}\to [0,1]$ with $f(u_1)=1/2$ is an element of $H$ so that $H$ lives on $\Omega$.
            We show that there exists a function $f \in H$ with domain $U$.
            For $f:V\to [0,1]$ and $g:W\to [0,1]$ in $H$, define $f\preccurlyeq g$ if $V\sqsubseteq W$ and $f=g$ restricted on $V$.
            Let now $(f_i)$ be a chain in $H$ and define $f:V:=\sqcup V_i\to [0,1]$, $u=\sum u_j|A_j\mapsto f(u)=\sum f_j(u_j)|A_j$ where $u_j\in V_j$ for every $j$ is a well-defined conditional function in $H$.
            Indeed, $H$ is a conditional set, the $f_i$ are restrictions of each others and $V_i\sqsubseteq V_j$ if $f_i\preccurlyeq f_j$.
            By Zorn's Lemma, there exists a maximal function $f \in H$, $f:V\to [0,1]$.
            Next we show that $V=U$.
            For the sake of contradiction, suppose that $V=\set{u_k: 1\leqslant k\leqslant n}$ for some $n \in \mathbf{N}$ on some non-trivial event $A$.
            Without loss of generality, assume that $A=\Omega$.
            Define $g:\set{u_k: 1\leqslant k\leqslant n+1}\to [0,1]$ by setting $g=f$ on $\set{u_k: 1\leqslant k\leqslant n}$ and 
            \begin{equation*}
                g(u_{n+1})=\frac{\displaystyle \sup_{l\leqslant n} \Set{f(u_l): u_l <u_{n+1}}+\inf_{l\leqslant n}\Set{f(u_l): u_{n+1} < u_l}}{2}
            \end{equation*}
            As for those $n\leqslant m\leqslant n+1$, it follows that $m=n|A+(n+1)|A^c$ and we set $g(u_m)=g(u_n)|A+g(u_{n+1})|A^c$.
            By construction, $g:\set{u_k: 1\leqslant k\leqslant n+1}\to[0,1]$ is an element of $H$ which coincides on $V$ with $f$.
            Since $g$ is defined on $V\sqcup \set{u_{n+1}}$ it contradicts the maximality of $f$.
            Thus the domain of the maximal function $f\in H$ is $U$. 
            
        \item Let $U=\{u_n\colon n\in \mathbf{N}\}$ and $f:U\to [0,1]$ as defined in the previous step.
            Suppose that $V,W\sqsubseteq U$ living on $\Omega$ satisfy
            \begin{enumerate}[label=(\alph*)]
                \item\label{assumption1} $V\sqcup W=U$,
                \item\label{assumption2} $V\leqslant W$,\footnote{In the sense that for all $s\in V$ and for all $t\in W$ it holds $s \leqslant t$.}
                \item\label{assumption3} $\emptyset$ is the unique condition on which $V$ and $W$ have both a maximum and a minimum, respectively.
            \end{enumerate}
            Then 
            \begin{equation*}
                \sup_{s \in V}f(s) = \inf_{t \in W}f(t).
            \end{equation*}
            By \ref{assumption2} it holds $\sup_{V} f(s) \leqslant \inf_{ W} f(t)$.
            In order to show the reverse inequality, according to \ref{assumption3}, it is sufficient to suppose that $V$ and $W$ have both nowhere a maximum and a minimum, respectively, since then the gap is the largest.
            For the sake of contradiction, up to conditioning, suppose that 
            \begin{equation*}
                \sup_{V}f(s)+ \varepsilon < \inf_{W} f(t)  
            \end{equation*}
            for some $\varepsilon>0$.
            Choose $s_0=u_m\in V$ and $t_0=u_n\in W$ such that 
            \begin{equation}\label{eq:eq01}
                \sup_Vf(s) - \varepsilon \leqslant f(s_0) \leqslant \sup_Vf(s) \quad \text{and} \quad 
                \inf_Wf(t) \leqslant f(t_0) \leqslant \inf_Wf(t) + \varepsilon.
            \end{equation}
            Since $V$ has nowhere a maximum, there exists $u_k\in V$ such that $u_m< u_k< u_n$ and $k>n,m$.
            Let $k^\prime=\min\set{k>n,m: u_n<u_k<u_m}$. 
            By construction of $f$ and since \ref{assumption1} holds, it follows that
            \begin{equation*}
                f(u_{k^\prime})=\frac{f(u_n)+f(u_m)}{2}=\frac{f(s_0)+ f(t_0)}{2}.
            \end{equation*}
            Adding both inequalities in \eqref{eq:eq01} yields
            \begin{equation*}
                \frac{\sup_Vf(s) + \inf_Wf(t)}{2} -\frac{\varepsilon}{2} \leqslant f(u_{k^\prime}) \leqslant \frac{\sup_Vf(s) + \inf_Wf(t)}{2} +\frac{\varepsilon}{2},
            \end{equation*}
            and therefore $\sup_Vf(s)<f(u_{k^\prime})<\inf_Wf(t)$ contradicting $u_{k^\prime} \in V$.
        \item Define $g:S\to \mathbf{R}$ by $g(s)=\sup_{u\in U, u\leqslant s} f(u)$.
            By construction, $g$ is a conditionally strictly increasing extension of $f$ since $U$ is a conditionally countable order dense subset of $S$.
            Let $(s,t)$ be a conditional gap of $g(S)$ and $A$ be the maximal event such that $s<t$, that is $(s,t)=(s,t)|A+\set{s}|B$.
            Without loss of generality, suppose that $A=\Omega$ and $B=\emptyset$.
            Define $V=\set{u\in U: f(u)\leqslant s}$ and $W=\set{u\in U: f(u) \geqslant t}$.
            Since $s < t$ and $V$, $W$ satisfy \ref{assumption1} and \ref{assumption2} of the previous step, \ref{assumption3} has to be violated.
            Hence $V$ and $W$ have both a conditional maximum and minimum respectively on some maximal non-empty event $C$, that is, $s=f(u_n)$ and $t=f(u_m)$ on $C$ for some $n,m$.
            Thus $s|C, t|C\in g(S)|C$ showing that $(s,t)|C=]s,t[|C$.
            If $C$ is non equal to $\Omega$, we follow the same argumentation but conditioned on $C^c$ which yields a contradiction with the maximality of $C$.
            Therefore, $(s,t)=]s,t[$ which ends the proof.
    \end{enumerate}
\end{proof} 
\begin{theorem}\label{thm:debreu}
    Any numerically representable conditionally upper semi-continuous preference order admits a conditionally upper semi-continuous numerical representation.
\end{theorem}
\begin{proof}
    Let $\tilde{U}$ be a numerical representation of a conditionally upper semicontinuous preference order $\succcurlyeq$. 
    According to Debreu's Gap Lemma \ref{gaptheorem} there exists a conditional function $g:Im(\tilde{U})\to \mathbf{R}$ such that all the conditional gaps $(s,t)$ of $g(Im(\tilde{U}))$ are of the following form
    \begin{equation*}
        (s,t)=s|A+]s,t[|B, \quad \text{for }s<t.
    \end{equation*}
    Since $g$ is strictly increasing, it follows that $U=g\circ\tilde{U}$ is a conditional numerical representation of $\succcurlyeq$ as well.
    Clearly $Im(U)=g(Im(\tilde{U}))$.
    In order to verify the upper semi-continuity, pick $m \in \mathbf{R}$.
    We distinguish between the following cases:
    \begin{itemize}
        \item If $m=U(y)$, then $\set{x \in X: U(x)\geqslant m}=\set{x \in X: U(x)\geqslant U(y)}=\set{x \in X: x\succcurlyeq y}$ which is conditionally closed by assumption.
        \item If $m \in ]s,t[$ where $]s,t[$ is a conditional gap of $Im(U)$, then $t=U(y)$ for some $y \in X$, and thus $\set{x \in X: U(x)\geqslant m}=\set{x\in X: U(x)\geqslant t}=\set{x \in X: U(x)\geqslant U(y)}=\set{x \in X: x\succcurlyeq y}$ which is also conditionally closed by assumption.
        \item If $m=s$ where $\set{s}$ is a conditional gap of $Im(U)$, then let $(s_n)=(U(y_n)) \sqsubseteq Im(U)$ be a conditional sequence such that $s_n \nearrow s$.
            It holds $\set{x \in X: U(x)\geqslant s}=\sqcap_n \set{x \in X: U(x)\geqslant s_n}=\sqcap_n\set{x \in X: U(x)\geqslant U(y_n)}=\sqcap_n \set{x \in X: x\succcurlyeq y_n}$ which is conditionally closed as the conditional intersection of closed sets.
    \end{itemize}
    Since $\mathbf{R}=Im(U)\sqcup [Im(U)]^\sqsubset$ and $[Im(U)]^\sqsubset$ is made of gaps of the form $(s,t)=\set{s}|A+]s,t[|B$, it follows that any $m \in \mathbf{R}$ belongs conditionally to one of the previous three cases.
    Thus $U$ is conditionally upper semi-continuous. 
\end{proof}

\begin{appendix}

\section{Technical Proofs}\label{appendix:01}
\subsection{Proof of Theorem \ref{thm:rader}.}
\begin{proof}[Proof of Theorem \ref{thm:rader}]
    Let $\mathcal{O}=(O_n)_{n \in \mathbf{N}}$ be a conditionally countable topological base of $X$ and $\mu$ be a strictly positive measure on $\mathbf{N}$.
    We know that $Z(x):=\mathcal{U}(x)^{\sqsubset}$ is conditionally open for every $x\in X$. 
    Fix some $x\in X$ and let $A$ be the event on which lives $Z(x)$. 
    Then $\set{n\in \mathbf{N}: O_n|A\sqsubseteq Z(x)}$ is a conditional subset of $\mathbf{N}$.
    Next define
    $$
    U(x)= \mu(\set{n\in\mathbf{N}: O_n|A\sqsubseteq Z(x)})|A+ 0|A^c.
    $$
    If $x \succcurlyeq y$, then $U(x)\geqslant U(y)$ since $Z(y)\sqsubseteq Z(x)$.
    Otherwise if $x \succ y$, then $y \in Z(x)$.
    Since $Z(x)$ is conditionally open, there exists a neighborhood $O_{i_0}$ of $y$ such that $O_{i_0} \sqsubseteq Z(x)$.
    However, since $y \in Z(y)^\sqsubset \sqcap O_{i_0}$, it follows that $O_{i_0}$ is nowhere a subset of $Z(y)$.
    Hence, $U(x)\geqslant U(y)+\mu(\set{i_0})>U(y)$.
    By Theorem \ref{thm:debreu} we can choose $U$ to be conditionally upper semi-continuous which ends the proof. 
\end{proof}

\subsection{Proof of Theorem \ref{thm:affine}.}
We will follow the classical proof adapted to the conditional setting.
\begin{lemma}\label{lem1}
    Let $\succcurlyeq$ be a conditionally complete preference order satisfying both the conditional independence and Archimedean axioms. 
    Then the following assertions hold:
    \begin{enumerate}[label=(\roman*)]
        \item\label{prop1} If $x\succ y$, then $\beta x + (1-\beta)y \succ \alpha x + (1-\alpha) y$ for all $0\leqslant \alpha<\beta\leqslant 1$.
        \item\label{prop2} If $x \succ z$ and $x\succcurlyeq y \succcurlyeq z$, then there exists a unique $\alpha \in [0,1]$ with $y\sim \alpha x+(1-\alpha)z$.
        \item\label{prop3} If $x \sim y$, then $\alpha x+(1-\alpha)z\sim \alpha y+(1-\alpha)z$ for all $\alpha\in [0,1]$ and all $z\in \mathcal{M}$.
    \end{enumerate}
\end{lemma}
\begin{proof}
    \begin{enumerate}[label=(\roman*), fullwidth]
        \item Strictly analogous to the classical proof, see for instance \cite[p. 54]{foellmer01}. 
        \item Up to conditioning, we may assume that $x\succ y \succ z$.\footnote{On the condition where $x\sim y$ set $\alpha=1$ and on the condition where $y\sim z$, set $\alpha=0$.}
            The candidate is 
            \begin{equation*}
                \alpha:=\sup\set{\beta\in [0,1]:y\succcurlyeq \beta x + (1-\beta)z}
            \end{equation*} 
            We obtain a partition $A,B,C$ of $\Omega$ such that $y\sim \alpha x+(1-\alpha)z$ on $A$, $y\succ \alpha x+(1-\alpha)z$ on $B$ and $\alpha x+(1-\alpha)z\succ y$ on $C$.
            Conditioned on $B$ and $C$ respectively, we may apply the classical argumentation, see for instance \cite[p. 54]{foellmer01} yielding a contradiction showing that $B=C=\emptyset$ and therefore $A=\Omega$ which ends the proof.
            As for the uniqueness, this is a consequence of the first point.
        \item Let $\alpha \in [0,1]$ and $z \in X$.
            There exists a partition of $A,B,C$ of $\Omega$ such that $\alpha x+(1-\alpha)z\sim \alpha y+(1-\alpha)z$ on $A$, $\alpha x+(1-\alpha)z\succ \alpha y+(1-\alpha)z$ on $B$ and $\alpha x+(1-\alpha)z\succ \alpha y+(1-\alpha)z$ on $C$.
            The same contradiction argumentation as in the classical case, \citep[p. 54--55]{foellmer01}, conditioned on $B$ and $C$, respectively, shows that $B=C=\emptyset$ and therefore $A=\Omega$.
    \end{enumerate}
\end{proof}
\begin{proof}[Theorem \ref{thm:affine}]
    Let $x,y \in X$ be such that $x\succ y$ and define the conditional convex subset $N_{x,y}:=\set{z \in X: x\succcurlyeq z\succcurlyeq y}$.
    For $z\in N_{x,y}$, part \ref{prop2} of Lemma \ref{lem1} yields a unique $\alpha\in[0,1]$ such that $z\sim \alpha x + (1-\alpha)y$.
    Setting $U(z):=\alpha$, $z \in N_{x,y}$ provides a well-defined conditional function from $N_{x,y}$ to $[0,1]$. 
    Indeed, let $[a_i,z_i]\subseteq \mathcal{A}\times N_{x,y}$, and denote $\alpha_i=U(z_i)$ and $\alpha=U(\sum z_i|A_i)$.
    There exists a partition $A,B,C$ of $\Omega$ such that $\alpha= \sum \alpha_i|A_i$ on $A$, $\alpha > \sum  \alpha_i|A_i$ on $B$ and $\sum \alpha_i|A_i >\alpha$ on $C$.
    In particular, $\alpha>\alpha_i$ on $B\cap A_i$ and $\alpha_i>\alpha$ on $C\cap A_i$ for every $i$.
    Hence, if either $B$ or $C$ were non-empty events, this would contradict the uniqueness of some $\alpha_i$ on $B\cap A_i\neq \emptyset$ or $C\cap A_i\neq \emptyset$.
    Hence $B=C=\emptyset$ showing that $A=\Omega$.
    The extension to $X$ follows exactly the same argumentation as the classical case, see \citep[p. 55]{foellmer01}.
\end{proof}

\subsection{Proof of Proposition \ref{thm:repweak}.}
In a conditional topological space $X$ with conditional topological dual $X^\ast$, the conditional absolute polar of a set $O\sqsubseteq X$ living on $\Omega$ is given by
\begin{equation*}
    O^\bullet =\Set{x^\ast \in X^\ast: \abs{\langle x^\ast, x\rangle}\leqslant 1 \text{ for all }x \in O}
\end{equation*}
\begin{lemma}\label{lem:abdeckung}
    Let $X$ be a conditional topological vector space with conditional dual $X^\ast$ and $\mathcal{O}$ a conditional base of neighborhoods of $0$ in $X$.
    Then 
    \begin{equation*}
        X^\ast=\bigsqcup_{O\in \mathcal{O}} O^\bullet.
    \end{equation*}
\end{lemma}
\begin{proof}
    Let $x^\ast \in X^\ast$, then $V=[x^\ast]^{-1}([-1,1])$ is a conditional neighborhood of $0$.
    In particular, $x^\ast \in V^\bullet$.
    Choose $O\in\mathcal{O}$ such that $O\sqsubseteq V$.
    Then $V^\bullet \sqsubseteq O^\bullet$, and thus $x^\ast \in O^\bullet$.
    The reciprocal is immediate since $O^\bullet \sqsubseteq X^\ast$.
\end{proof}
\begin{proposition}\label{prop:metrizability}
    Let $X$ be a locally convex conditional topological vector space $X$ which is conditionally separable.
    Relative to any conditionally $\sigma(X^\ast, X)$-compact subset $C\sqsubseteq X$, the $\sigma(X^\ast, X)$-topology is conditionally metrizable.
\end{proposition}
\begin{proof}
    Without loss of generality, by the conditional version of the Banach-Alaoglu Theorem \citep{djkk2013} and the previous lemma, we may assume that $C=O^\bullet$ for some conditional neighborhood $O$ of $0$ in $X$.
    
    First, we construct a conditional distance on $O^\bullet$ as follows.
    Let $(x_n)\sqsubseteq O$ be a conditionally dense sequence in $O$ and define $d:O^\bullet \times O^\bullet\to \mathbf{R}_+$ by
    \begin{equation}
        d(x^\ast,y^\ast)=\sum_{n \in \mathbf{N}} \frac{1}{2^n}\frac{ \abs{ \langle x^\ast-y^\ast,x_n\rangle }}{1+\abs{\langle x^\ast-y^\ast , x_n\rangle}},\quad x^\ast, y^\ast \in O^\bullet.
    \end{equation}
    Straightforward inspection shows that it is a well-defined conditional function and a translation invariant distance on $O^\bullet$.
    Indeed, as a locally convex conditional topological vector space, $X$ separates the points of $X^\ast$ and a fortiori those of $O^\bullet$.
    Furthermore,
    \begin{equation*}
        \Set{x^\ast \in X^\ast: \abs{\langle x^\ast, x_k\rangle }<r, 1\leqslant k\leqslant n }\sqsubseteq \Set{x^\ast \in X^\ast: d(0,x^\ast)<r+2^{-n+1}}=B_{r+2^{-n+1}}(0),
    \end{equation*}
    for every $n\in \mathbf{N}$ and $r>0$.
    This shows that the conditional topology generated by $d$ on $O^\bullet$ is weaker than $\sigma(X^\ast,X)$, that is, $\tau_d\sqsubseteq \sigma(X^\ast,X)$.\footnote{Note that the construction of such a metric can be done similarly on $X^\ast$ by considering a conditionally dense sequence $(x_n)$ of elements in $X$, and the topology induced by $d$ on $X^\ast$ is therefore weaker that $\sigma(X^\ast,X)$.}

    Second, we show that these topologies coincide. 
    To this end, we consider the identity map $Id:(O^\bullet,\sigma(X^\ast,X))\to (O^\bullet ,d)$ which is a bijection.
    Let $(x^\ast_\alpha)\sqsubseteq O^\bullet$ be a conditional net converging in $\sigma(X^\ast,X)$ to $x^\ast\in O^\bullet$.
    For $r>0$, choose $k\in \mathbf{N}$ such that $\sum_{n>k} 2^{-n}< r$.
    Since $x_\alpha^\ast, x^\ast \in O^\bullet$ and $x_n\in O$, it follows that $\abs{\langle x^\ast_\alpha-x^\ast, x_n\rangle} \leqslant 2$ for every $n \in\mathbf{N}$.
    Hence,
    \begin{equation}
        d(x_\alpha,x)\leqslant \sum_{1\leqslant n\leqslant k}\abs{\langle x^\ast_\alpha-x^\ast, x_n\rangle }+2r
    \end{equation}
    Since $\abs{\langle x^\ast_\alpha-x^\ast, x_n\rangle }\to 0$ for every $n \in \mathbf{N}$, it follows that $\limsup d(x_\alpha,x)\leqslant 2r$ for every $r>0$.
    This shows that $Id$ is continuous.
    Now, $(O^\bullet,\sigma(X^\ast,X))$ is conditionally compact due to the conditional version of Banach-Alaoglu and $(C,d)$ is conditionally Hausdorff, it follows that $Id$ is conditionally bi-continuous.\footnote{Every conditionally continuous bijection $f:C\to D$ where $C$ is a conditionally compact and $D$ conditionally Hausdorff is conditionally bi-continuous. Indeed, every conditionally closed set $F\sqsubseteq C$ is conditionally compact and since $f$ is conditionally continuous, it follows that $f(F)$ is conditionally compact, see \citep[Proposition 3.35]{djkk2013}. Moreover, since $D$ is conditionally Hausdorff, it holds $f(F)$ is conditionally closed.}
    Hence, $V\in \tau_d$ for every $V\in \sigma(X^\ast,X)$ showing that $\tau_d=\sigma(X^\ast,X)$ relative to $O^\bullet$.
\end{proof}

\begin{proof}[of Proposition \ref{thm:repweak}]
    Denoting by $(O_n)$ the conditional countable neighborhood of $0$ in $X$, define $\succcurlyeq_n$ as the conditional restriction to $O_n^\bullet$ of $\succcurlyeq$.
    Clearly, $\succcurlyeq_n$ is conditionally $\sigma(X^\ast,X)$-upper semi-continuous for every $n$.
    Furthermore, by the conditional version of Banach-Alaoglu, $O_n^\bullet$ is $\sigma(X^\ast,X)$-compact.
    By Proposition \ref{prop:metrizability}, it follows that $O_n^\bullet$ is conditionally metrizable and compact, hence conditionally second countable.
    Theorem \ref{thm:rader} implies that $\succcurlyeq_n$ is representable and by Theorem \ref{thm:numrep01}, it admits a conditionally countable order dense subset $Z_n\sqsubseteq O_n^\bullet$.
    By means of \citep[Lemma 2.33]{djkk2013}, $Z:=\sqcup Z_n$ is a conditionally countable set.
    By means of Lemma \ref{lem:abdeckung}, straightforward inspection shows that $Z$ is $\succcurlyeq$-conditionally order dense.
    Therefore, once again by means of Theorem \ref{thm:numrep01}, $\succcurlyeq$ admits a conditional numerical representation.
    Theorem \ref{thm:debreu} guarantees that such a conditional numerical representation can be chosen $\sigma(X^\ast, X)$-upper semi-continuous.
\end{proof}

\subsection{Proof of the Automatic Continuity Result.}

\begin{proposition}\label{prop:automatic_continuity}
    Let $X$ be a conditional Fr\'echet lattice and $Z\sqsubseteq X$ be conditionally monotone and convex.
    If $f^{-1}(Z)$ is conditionally closed in $[0,1]$ for every given pair $x,y \in X$, where $f:[0,1]\to X$, $\alpha\mapsto \alpha x+(1-\alpha)y$, then $Z$ is conditionally closed in $X$.
\end{proposition}
\begin{proof}
    Denote by $d$ the conditional Fr\'echet distance on $X$, and let $(x_n)$ a conditional sequence of elements in $Z$ conditionally converging to $x \in X$.
    Up to a rapid conditional subsequence, we may suppose that $d(x_n,x)\leq 2^{-n}/n$, $n\in \mathbf{N}$.
    It follows that $\sum_{k\geqslant 1} k (x_k-x)^+$ is conditionally converging.
    Indeed, since the conditional Fr\'echet distance respects the conditional absolute value, it follows for $n< m$
    \begin{multline}
        d\left(\sum_{1\leqslant k \leqslant n} k (x_k-x)^+,\sum_{1\leqslant k\leqslant m} k (x_k-x)^+\right)\leqslant d\left(0,\sum_{n< k\leqslant m} k (x_k-x)\right)\\
        \leqslant \sum_{n<k\leqslant m} kd(x,x_k)\leqslant \sum_{n<k\leqslant m} 2^{-k}\xrightarrow[m,n\to \infty]{}0
    \end{multline}
    Hence, the conditional completeness of $X$ implies that $y=x+\sum_{k\geqslant 1} k (x_k-x)^+$ is well-defined and for each $\alpha \in [0,1[$ it holds 
    \begin{equation*}
        \alpha x+(1-\alpha)y=x+(1-\alpha)\sum_{k\geqslant 1} k(x_k-x)^+\geqslant x+(1-\alpha) n(x_n-x)^+
    \end{equation*}
    for every $n\in\mathbf{N}$.
    Choosing $n\geqslant 1/(1-\alpha)$ yields
    \begin{equation*}
        \alpha x+(1-\alpha)y\geqslant x+(1-\alpha) n(x_n-x)^+\geqslant x+(x_n-x)^+\geqslant x_n \in Z
    \end{equation*}
    By monotonicity of $Z$, it holds $\alpha x+(1-\alpha)y \in Z$.
    Since $n$ can be chosen arbitrarily large, it follows that $[0,1[\sqsubseteq f^{-1}(Z)$ for $\alpha \mapsto x+(1-\alpha)y$, $\alpha \in [0,1]$.
    By assumption, the latter set is conditionally closed in $[0,1]$, therefore $1 \in f^{-1}(Z)$, that is, $x \in Z$ ending the proof.
\end{proof}

\begin{proof}[of Proposition \ref{prop:aut:cont}]
    Fix an $x_0\in X$ and denote by $Z:=\mathcal{U}(x_0)$.
    Then $Z$ is conditionally convex and monotone since $\succcurlyeq$ is so.
    We show that $I:=f^{-1}(Z)$ is conditionally closed in $[0,1]$ where $f:[0,1]\to X$, $\alpha\mapsto \alpha x+(1-\alpha)y$ for any given $x,y \in X$.
    Up to conditioning, we may assume that $I$ lives on $\Omega$ -- in particular $I$ is not conditionally empty\footnote{Recall that the conditional emptyset is conditionally closed.}.
    Since $Z$ is conditionally convex and $f$ conditionally affine, it follows that $I$ is conditionally convex.
    Therefore, $I$ is an interval $(s,t)\sqsubseteq [0,1]$ where $s = \inf I$ and $t=\sup I$.
    If $s=t$, then $I$ is a singleton and therefore is conditionally closed.
    Otherwise, let $s<t$ without of loss of generality. 
    Suppose now, for the sake of contradiction, $s \in I^\sqsubset$.
    That is to say, $\alpha x+(1-\alpha)x_0 \succcurlyeq y\succ sx +(1-s)x_0$ for every $\alpha \in ]s,t[$.
    The one-sided Archimedean axiom yields a $\beta \in ]0,1[$ such that
    \begin{equation}\nonumber
        y\succ \beta (\alpha x+(1-\alpha)z)+(1-\beta)(s x+(1-s)z)=\gamma x+(1-\gamma)z
    \end{equation}
    where $\gamma=\beta \alpha+(1-\beta)s$.
    Since $\beta>0$ and $s<\alpha$, it follows that $\gamma>s$ contradicting the definition of $s$.
    Hence, $sx+(1-s)z \in Z$, and thus $s \in (s,t)$.
    By an analogous argumentation for $t$, it can be concluded that $(s,t)=[s,t]$.
    By Proposition \ref{prop:automatic_continuity}, it holds $Z$ is conditionally closed which ends the proof.
\end{proof}

\end{appendix}
\bibliographystyle{abbrvnat}
\bibliography{bibliography}

\begin{thebibliography}{42}
\providecommand{\natexlab}[1]{#1}
\providecommand{\url}[1]{\texttt{#1}}
\expandafter\ifx\csname urlstyle\endcsname\relax
  \providecommand{\doi}[1]{doi: #1}\else
  \providecommand{\doi}{doi: \begingroup \urlstyle{rm}\Url}\fi

\bibitem[Acciaio et~al.(2011)Acciaio, F{\"o}llmer, and Penner]{penner02}
B.~Acciaio, H.~F{\"o}llmer, and I.~Penner.
\newblock Risk assessment for uncertain cash flows: Model ambiguity,
  discounting ambiguity, and the role of bubbles.
\newblock \emph{Forthcoming in Finance and Stochastics}, 2011.

\bibitem[Aumann(1962)]{aumann01}
R.~J. Aumann.
\newblock Utility theory without the completeness axiom.
\newblock \emph{Econometrica}, 30\penalty0 (3):\penalty0 445--462, 1962.

\bibitem[Bewley(2001)]{bewley2001}
T.~F. Bewley.
\newblock Knightian decision theory. {P}art {I}.
\newblock \emph{Decisions in Economics and Finance}, 2001.

\bibitem[Bielecki et~al.(2014)Bielecki, Cialenco, Drapeau, and
  Karliczek]{drapeau2014}
T.~R. Bielecki, I.~Cialenco, S.~Drapeau, and M.~Karliczek.
\newblock Dynamic assessment indices.
\newblock \emph{Stochastics. An International Journal of Probability and
  Stochastic Processes (forthcoming)}, 2014.

\bibitem[Cerreia-Vioglio et~al.(2011)Cerreia-Vioglio, Maccheroni, Marinacci,
  and Montrucchio]{marinacci03}
S.~Cerreia-Vioglio, F.~Maccheroni, M.~Marinacci, and L.~Montrucchio.
\newblock Uncertainty averse preferences.
\newblock \emph{Journal of Economic Theory}, 146\penalty0 (4):\penalty0
  1275--1330, 2011.

\bibitem[Cheridito and Kupper(2009)]{ck09}
P.~Cheridito and M.~Kupper.
\newblock Recursivity of indifference prices and translation-invariant
  preferences.
\newblock \emph{Mathematics and Financial Economics}, 2:\penalty0 173--188,
  2009.

\bibitem[Cheridito et~al.(2006)Cheridito, Delbaen, and Kupper]{cheridito01}
P.~Cheridito, F.~Delbaen, and M.~Kupper.
\newblock Dynamic monetary risk measures for bounded discrete-time processes.
\newblock \emph{Electronic Journal of Probability}, 11(3):\penalty0 57--106,
  2006.

\bibitem[Cialenco et~al.(2010)Cialenco, Bielecki, and Zhang]{cialenco2010}
I.~Cialenco, T.~R. Bielecki, and Z.~Zhang.
\newblock Dynamic coherent acceptability indices and their applications to
  finance.
\newblock \emph{To appear in Mathematical Finance}, 2010.

\bibitem[Debreu(1954)]{debreu01}
G.~Debreu.
\newblock Representation of a preference ordering by a numerical function.
\newblock In C.~C. Thrall, R.M. and R.~Davis, editors, \emph{Decision Process},
  pages 159--165. John Wiley, New York, 1954.

\bibitem[Debreu(1964)]{Debreu03}
G.~Debreu.
\newblock Continuity properties of {P}aretian utility.
\newblock \emph{International Economic Review}, 5\penalty0 (3):\penalty0
  285--293, 1964.

\bibitem[Detlefsen and Scandolo(2005)]{detlefsen01}
K.~Detlefsen and G.~Scandolo.
\newblock Conditional and dynamic convex risk measures.
\newblock \emph{Finance and Stochastics}, 9:\penalty0 539--561, 2005.

\bibitem[Dillenberger et~al.(2014)Dillenberger, Lleras, Sadowski, and
  Takeoka]{dillenberger2014}
D.~Dillenberger, J.~S. Lleras, P.~Sadowski, and N.~Takeoka.
\newblock A theory of subjective learning.
\newblock \emph{Journal of Economic Theory}, 153:\penalty0 287--312, 2014.
\newblock ISSN 0022-0531.

\bibitem[Drapeau and Kupper(2013)]{drapeau01}
S.~Drapeau and M.~Kupper.
\newblock Risk preferences and their robust representation.
\newblock \emph{Mathematics of Operations Research}, 28\penalty0 (1):\penalty0
  28--62, 2013.

\bibitem[Drapeau et~al.(2013)Drapeau, Jamneshan, Karliczek, and
  Kupper]{djkk2013}
S.~Drapeau, A.~Jamneshan, M.~Karliczek, and M.~Kupper.
\newblock The algebra of conditional sets and the concepts of conditional
  topology and compactness.
\newblock \emph{Journal of Mathematical Analysis and Applications
  (forthcoming)}, 2013.

\bibitem[Dubra and Ok(2002)]{dubra2002}
J.~Dubra and E.~A. Ok.
\newblock A model of procedural decision making in the presence of risk.
\newblock \emph{International Economic Review}, 43\penalty0 (4):\penalty0
  1053--1080, 2002.

\bibitem[Dubra et~al.(2004)Dubra, Maccheroni, and Ok]{dubra2004}
J.~Dubra, F.~Maccheroni, and E.~A. Ok.
\newblock Expected utility theory without the completeness axiom.
\newblock \emph{Journal of Economic Theory}, 115\penalty0 (1):\penalty0
  118--133, 2004.

\bibitem[Duffie and Epstein(1992)]{epstein03}
D.~Duffie and L.~G. Epstein.
\newblock Stochastic differential utility.
\newblock \emph{Econometrica}, 60\penalty0 (2):\penalty0 353--94, 1992.

\bibitem[Eliaz and Ok(2006)]{eliaz2006}
K.~Eliaz and E.~A. Ok.
\newblock Indifference or indecisiveness? {C}hoice-theoretic foundations of
  incomplete preferences.
\newblock \emph{Games and Economic Behavior}, 56\penalty0 (1):\penalty0 61 --
  86, 2006.

\bibitem[Epstein and Zin(1989)]{epstein02}
L.~G. Epstein and S.~E. Zin.
\newblock Substitution, risk aversion, and the temporal behavior of consumption
  and asset returns: A theoretical framework.
\newblock \emph{Econometrica}, 57\penalty0 (4):\penalty0 937--69, 1989.

\bibitem[Evren and Ok(2011)]{evren2011}
O.~Evren and E.~A. Ok.
\newblock On the multi-utility representation of preference relations.
\newblock \emph{Journal of Mathematical Economics}, 47\penalty0 (4-5):\penalty0
  554--563, 2011.

\bibitem[Filipovic et~al.(2009)Filipovic, Kupper, and Vogelpoth]{kupper03}
D.~Filipovic, M.~Kupper, and N.~Vogelpoth.
\newblock Separation and duality in locally {$L^0$}-convex modules.
\newblock \emph{Journal of Functional Analysis}, 256:\penalty0 3996--4029,
  2009.

\bibitem[Filipovic et~al.(2011)Filipovic, Kupper, and
  Vogelpoth]{fillipovic2012}
D.~Filipovic, M.~Kupper, and N.~Vogelpoth.
\newblock Approaches to conditional risk.
\newblock \emph{SIAM Journal of Financial Mathematics}, 2011.

\bibitem[F\"{o}llmer and Schied(2004)]{foellmer01}
H.~F\"{o}llmer and A.~Schied.
\newblock \emph{Stochastic Finance. An Introduction in Discrete Time}.
\newblock de Gruyter Studies in Mathematics. Walter de Gruyter, Berlin, New
  York, 2 edition, 2004.

\bibitem[Fritelli and Maggis(2011)]{fritelli04}
M.~Fritelli and M.~Maggis.
\newblock Conditional certainty equivalent.
\newblock \emph{International Journal of Theoretical and Applied Finance},
  14\penalty0 (1):\penalty0 41--59, 2011.

\bibitem[Givant and Halmos(2009)]{halmos2009}
S.~Givant and P.~Halmos.
\newblock \emph{Introduction to Boolean algebras}.
\newblock Springer, 2009.

\bibitem[Jamneshan et~al.(2015)Jamneshan, Kupper, and Streckfuss]{asgar2015}
A.~Jamneshan, M.~Kupper, and M.~Streckfuss.
\newblock Conditional measures and integrals.
\newblock \emph{In preparation}, 2015.

\bibitem[Karni(1993{\natexlab{a}})]{karni1993}
E.~Karni.
\newblock Subjective expected utility theory with state-dependent preferences.
\newblock \emph{Journal of Economic Theory}, 60\penalty0 (2):\penalty0
  428--438, 1993{\natexlab{a}}.

\bibitem[Karni(1993{\natexlab{b}})]{karni1993b}
E.~Karni.
\newblock A definition of subjective probabilities with state-dependent
  preferences.
\newblock \emph{Econometrica}, 61\penalty0 (1):\penalty0 pp. 187--198,
  1993{\natexlab{b}}.

\bibitem[Kreps and Porteus(1978)]{kreps02}
D.~M. Kreps and E.~L. Porteus.
\newblock Temporal resolution of uncertainty and dynamic choice theory.
\newblock \emph{Econometrica}, 46\penalty0 (1):\penalty0 185--200, 1978.

\bibitem[Kreps and Porteus(1979)]{kreps03}
D.~M. Kreps and E.~L. Porteus.
\newblock Temporal von {N}eumann-{M}orgenstern and induced preferences.
\newblock \emph{Journal of Economic Theory}, 20\penalty0 (1):\penalty0 81--109,
  1979.

\bibitem[Luce and Krantz(1971)]{luce1971}
R.~D. Luce and D.~H. Krantz.
\newblock Conditional expected utility.
\newblock \emph{Econometrica}, 39\penalty0 (2):\penalty0 pp. 253--271, 1971.

\bibitem[Maccheroni et~al.(2006)Maccheroni, Marinacci, and
  Rustichini]{marinacci02}
F.~Maccheroni, M.~Marinacci, and A.~Rustichini.
\newblock Dynamic variational preferences.
\newblock \emph{Journal of Economic Theory}, 127\penalty0 (1):\penalty0 4--44,
  2006.

\bibitem[Ouwehand(2010)]{ouwehand2010}
P.~Ouwehand.
\newblock A simple proof of {D}ebreu's {G}ap {L}emma.
\newblock \emph{ORiON}, 2010.

\bibitem[Peleg(1970)]{peleg1970}
B.~Peleg.
\newblock Utility functions for partially ordered topological spaces.
\newblock \emph{Econometrica}, 38\penalty0 (1):\penalty0 93--96, 1970.

\bibitem[Piermont et~al.(2015)Piermont, Takeoka, and Teper]{piermont2015}
E.~Piermont, N.~Takeoka, and R.~Teper.
\newblock Learning the krepsian state: Exploration through consumption.
\newblock \emph{Preprint}, 2015.

\bibitem[Rader(1963)]{rader01}
T.~Rader.
\newblock The existence of a utility function to represent preferences.
\newblock \emph{The Review of Economic Studies}, 30\penalty0 (3):\penalty0
  229--232, 1963.

\bibitem[Richter(1966)]{richter1966}
M.~K. Richter.
\newblock Revealed preference theory.
\newblock \emph{Econometrica}, 34\penalty0 (3):\penalty0 635--645, 1966.

\bibitem[Robert(1987)]{aumann1987}
A.~Robert.
\newblock Letter from {R}obert {A}umann to {L}eonard {S}avage.
\newblock In J.~Dr\`eze, editor, \emph{Essays on Economic Decisions under
  Uncertainty}. Cambridge University Press, 1987.

\bibitem[Skiadas(1997{\natexlab{a}})]{skiadas01}
C.~Skiadas.
\newblock Conditioning and aggregation of preferences.
\newblock \emph{Econometrica}, 65\penalty0 (2):\penalty0 347--368,
  1997{\natexlab{a}}.

\bibitem[Skiadas(1997{\natexlab{b}})]{skiadas02}
C.~Skiadas.
\newblock Subjective probability under additive aggregation of conditional
  preferences.
\newblock \emph{Journal of Economic Theory}, 76\penalty0 (2):\penalty0
  242--271, 1997{\natexlab{b}}.

\bibitem[von Neumann and Morgenstern(1947)]{neumann01}
J.~von Neumann and O.~Morgenstern.
\newblock \emph{Theory of Games and Economics Behavior}.
\newblock Princeton University Press, 2nd edition, 1947.

\bibitem[Wakker(1987)]{wakker1987}
P.~Wakker.
\newblock Subjective probabilities for state dependent continuous utility.
\newblock \emph{Mathematical Social Sciences}, 14\penalty0 (3):\penalty0
  289--298, 1987.

\end{thebibliography}
\end{document}